\documentclass[11pt]{article}
\usepackage{graphicx}
\usepackage{pgfplots}

\usepackage[colorlinks=true,citecolor=blue]{hyperref}
\usepackage{natbib}
\usepackage{graphicx}
\usepackage{amsfonts}
\usepackage{amsmath}
\usepackage{amssymb}
\usepackage{url}
\usepackage{fancyhdr}
\usepackage{indentfirst}
\usepackage{enumerate}
\usepackage{titlesec}
\usepackage{amsthm}
\usepackage{subcaption}
\usepackage{dsfont}
\usepackage[misc]{ifsym}

\theoremstyle{definition}

\newtheorem{example}{Example}

\newtheorem{definition}{Definition}

\theoremstyle{plain}
\newtheorem{theorem}{Theorem}

\newtheorem{proposition}{Proposition}

\theoremstyle{remark}
\newtheorem{remark}{Remark}

\usepackage{cases}

\theoremstyle{definition}

\def\N{\mathbb{N}}
\def\P{\mathbb{P}}
\def\p{\mathbb{P}}
\def\E{\mathbb{E}}

\def\R{\mathbb{R}}

\def\X{\mathcal{X}}

\def\X{\mathcal{X}}

\def\d{\mathrm{d}}
\DeclareMathOperator*{\esssup}{ess\text{-}sup}
\DeclareMathOperator*{\essinf}{ess\text{-}inf}

\newcommand{\VaR}{\mathrm{VaR}}
\newcommand{\ES}{\mathrm{ES}}

\newcommand{\RVaR}{\mathrm{RVaR}}

\renewcommand{\(}{\left(}
\renewcommand{\)}{\right)}

\renewcommand{\]}{\right]}
\DeclareMathOperator*{\argmin}{arg\,min}
\DeclareMathOperator*{\argmax}{arg\,max}
\def\laweq{\buildrel \d \over =}

\usepackage[onehalfspacing]{setspace}

\usepackage{bm}
\usepackage{tikz-qtree}
\usepackage{tikz}

\def\id{\mathds{1}}

\title{Risk exchange under infinite-mean Pareto models}
\author{Yuyu Chen\thanks{Department of Economics, University of Melbourne,  Australia. \Letter~{\scriptsize\url{yuyu.chen@unimelb.edu.au}}} \and Paul Embrechts\thanks{RiskLab, Department of Mathematics \& ETH Risk Center, ETH Zurich,  Switzerland. \Letter~{\scriptsize\url{embrechts@math.ethz.ch}}}
\and Ruodu Wang\thanks{Department of Statistics and Actuarial Science, University of Waterloo,  Canada. \Letter~{\scriptsize\url{wang@uwaterloo.ca}}}
}
\date{\today}

\begin{document}
	\maketitle
	\begin{abstract} 
%We find the perhaps surprising inequality
%that the weighted average of independent and identically distributed Pareto random variables with infinite mean is larger than one such random variable in the sense of first-order stochastic dominance.
%This result holds for more general models including super-Pareto distributions, negative dependence, and triggering events.
% that the weighted average of negatively dependent   super-Pareto random variables, possibly triggered by catastrophic events,  is larger than one such random variable in the sense of first-order stochastic dominance. The class of super-Pareto distributions is  extremely heavy-tailed and it includes the class of  infinite-mean Pareto distributions. 
%The result continues to hold if the super-Pareto random variables are  negatively dependent and  caused by triggering events.
We study the optimal decisions and equilibria of agents who aim to minimize their risks by allocating their positions over extremely heavy-tailed (i.e., infinite-mean) and possibly dependent losses. The loss distributions of our focus are super-Pareto distributions, which include the class of extremely heavy-tailed Pareto distributions. Using a recent result on stochastic dominance, we show that for a portfolio of super-Pareto losses, non-diversification is preferred by decision makers equipped with well-defined and monotone risk measures. The phenomenon that diversification is not beneficial in the presence of super-Pareto losses is further illustrated by an equilibrium analysis in a risk exchange market. First, agents with super-Pareto losses will not share risks in a market equilibrium. Second, transferring losses from agents bearing super-Pareto losses to external parties without any losses may arrive at an equilibrium which benefits every party involved.  
%The empirical studies show that extremely heavy tails exist in real datasets. 

% \textbf{JEL classification}: C02; C62; G22. 

\textbf{Keywords}: super-Pareto distributions; diversification; risk exchange; equilibrium; risk measures. 
	\end{abstract}

 %\noindent\rule{\textwidth}{0.5pt}
\section{Introduction}\label{sec:1}

Over the last decades, the insurance industry has observed a rising trend in both the frequency and magnitude of huge losses caused by natural disasters and man-made catastrophes (e.g., \cite{embrechts1999extreme}). In quantitative risk management, Pareto distributions (or generalized Pareto distributions) have been widely used in modeling catastrophic losses (see \cite{MFE15}), mainly due to their unique role in Extreme Value Theory (EVT): By the 
 Pickands-Balkema-de Haan Theorem (\cite{P75} and \cite{BD74}),
  generalized Pareto distributions are the only possible non-degenerate limiting distributions of excess-of-loss random variables.  
  Although statistical and actuarial methods often focus on finite-mean EVT models, 
 data analyses in the literature show that the best-fitted models for many  catastrophic losses in various contexts do not have finite means. 
 %see  \cite{IJW09} for earthquake losses, \cite{rizzo2009new} for  wind catastrophic losses, and \cite{hofert2012statistical} for losses in nuclear power accidents. Besides catastrophic losses, losses caused by operational risks also lead to infinite-mean Pareto models; see \cite{moscadelli2004modelling} for an analysis of the loss data
%collected by the Basel committee.
At the end of this section, we collect many examples and related literature on infinite-mean Pareto-type models.

It is well known that infinite-mean models lead to intriguing phenomenon in risk management; see e.g., \cite{IJW11}. 
This paper focuses on the following question. 
\begin{enumerate}
\item[\textbf{Q.}]
Suppose that there is a pool of identically distributed \emph{extremely heavy-tailed} losses (i.e., infinite mean), possibly statistically dependent. 
Each agent (e.g., a reinsurance provider) needs to decide whether and how to diversify in this pool.
\emph{Without knowing the  preferences} of the agents,  
what can we say about the   optimal decisions and equilibria in a  multiple-agent setting? 
\end{enumerate} 
%This question will be used to motivate our model and main result. 
% The multiple-agent version of this question is addressed in the equilibrium model in Section \ref{app:r1}.

Assume for simplicity that the losses in the pool, denoted by $X_{1},\dots,X_{n}$, are independent and identically distributed (iid) Pareto losses with infinite mean. Let $\theta_{1},\dots,\theta_n\ge 0$ with $\sum_{i=1}^{n}\theta_i=1$ be the exposures of an agent over $X_{1},\dots,X_{n}$.
Our analysis is built on a recent result of 
\cite{CEW24} that gives
\begin{equation}\label{eq:*}
X_{1}\le_{\rm st}\theta_{1}X_{1}+\dots+\theta_{n}X_{n},
\end{equation}
where $\le_{\rm st}$ stands for \emph{first-order stochastic dominance}; that is, for two random variables $X$ and $Y$, $X\le_{\rm st}Y$ if $\p(X\le x)\ge \p(Y\le x)$ for all $x\in\R$. This inequality is shown to be strict when the agent holds at least two losses in their portfolio. Intuitively, the left-hand side of \eqref{eq:*}  is a portfolio concentrated on $X_1$, and the right-hand side is a diversified portfolio of $X_1,\dots,X_n$. Inequality \eqref{eq:*} implies that if the agent aims to minimize their default probability, the optimal decision is non-diversification; this observation is generalized in Section \ref{sec:4}.

 Besides iid Pareto losses with infinite mean, \eqref{eq:*} and its implications also hold for weakly negatively associated and identically distributed super-Pareto losses by Theorem 1 of \cite{CEW24}, {as well as many other loss models considered in \cite{I05}, \cite{ALO24}, \cite{CS24}, and \cite{M24}. We will focus on super-Pareto losses in this paper, while keeping in mind that the results work for any models satisfying \eqref{eq:*}.} 
  We briefly summarize in Section \ref{sec:2} the definitions of super-Pareto distributions and weak negative association. 
% Given the risk management implications of \eqref{eq:*},
 Section \ref{sec:extension} presents several generalizations of \eqref{eq:*} beyond weakly negatively associated super-Pareto models considered by \cite{CEW24}. Propositions \ref{prop:convolution} and \ref{prop:mixture} provide some high-level conditions for generalizing the marginal distributions and the dependence structure.  Proposition \ref{prop:tail} provides a model for losses that are super-Pareto only in the tail region and Proposition \ref{cor:random} considers a classic insurance model.

 We discuss in Section \ref{sec:4} useful decision models in the presence of extremely heavy-tailed losses and the implications of \eqref{eq:*} and related inequalities on the risk management decision of a single agent. First, although  \eqref{eq:*}  never holds for non-degenerate random variables with finite mean (see Proposition 2 in \cite{CEW24}), we show that a similar preference for non-diversification exists for truncated super-Pareto losses, as long as the thresholds are high enough (Theorem \ref{prop:bounded}). Second,
 for super-Pareto losses, the action of diversification increases the risk \emph{uniformly for all risk preferences}, such as VaR, expected utilities, and distortion risk measures, as long as the risk preferences are monotone and well-defined.  The increase of the portfolio risk is strict, and it provides an important implication in decision making: For an agent who faces   super-Pareto losses   and aims to minimize their risk by choosing a position across these losses, the optimal decision is to take only one of the super-Pareto losses (i.e., no diversification).   

We proceed to study the equilibria of a risk exchange market for super-Pareto losses  in Section \ref{app:r1}.  
 As individual agents do not benefit from diversification over super-Pareto losses,
 % diversification  in a risk exchange market where  super-Pareto losses are present, 
 we may expect that agents will not share their losses. Indeed, if each agent in the market is associated with an initial position in one super-Pareto loss, the agents will merely exchange the entire loss position instead of risk sharing in an equilibrium model (Theorem \ref{th:equil} (i)).  The situation becomes quite different if the agents with initial losses can transfer their losses to external parties.
 If the external agents have a stronger risk tolerance than the internal agents, both parties may benefit by transferring losses from the internal to the external agents (Theorem \ref{th:external} (ii)). In Proposition \ref{prop:equil-ES}, we show that agents prefer to share losses with finite mean among themselves;  this is in  sharp contrast to the case of super-Pareto losses. 
% The above results are  consistent with the observations made by \cite{IJW11} based on a different model.

   In Section \ref{sec:6}, some  examples are presented to illustrate the presence  of extremely heavy tails in two real datasets in which the phenomenon of inequality \eqref{eq:*} or its generalizations can be empirically observed. %We proceed to study the diversification effects of extremely heavy-tailed Pareto losses with different tail indices. 
   Section \ref{sec:7} concludes the paper.  Some background on risk measures is put in Appendix \ref{app:A}.
Proofs of all results are put in Appendix \ref{app:proof}.

% \subsection{Review of infinite-mean Pareto-type models}
% \label{sec:rem1}

 \textbf{Review of infinite-mean Pareto-type models.}
The key assumption of our paper is that losses follow super-Pareto distributions, which have infinite mean. Whereas statistical models with some divergent higher moments are ubiquitous throughout the risk management literature, the infinite mean case needs more specific motivation. For power-tail data, a standard approach for the estimation of the underlying tail parameters is the Peaks Over Threshold (POT) methodology from EVT; see \cite{EKM97}. Other estimation methods  include the classic Hill estimators (see \cite{EKM97}) and the log-log rank-size estimation (see \cite{ibragimov2015heavy}). It is known that the Hill estimator may be sensitive to the dependence in the data and small sample sizes, and the log-log rank-size estimation can be biased in small samples; see  \cite{gabaix2011rank} for an improved version of log-log rank-size estimation. Below we discuss some examples from the literature leading to extremely heavy-tailed Pareto models.\footnote{For these examples, it turns out that infinite-mean models yield a better overall fit than finite-mean ones, although one can never say for sure that any real world dataset is generated by an infinite-mean model.}

In the parameterization used in Section \ref{sec:2}, a tail parameter $\alpha\le1$ corresponds to an infinite-mean Pareto model. 
% We discuss some examples of extremely heavy-tailed Pareto distributions  in the literature; extra data examples are provided in Section \ref{sec:62}.\footnote{These examples show that, at least, we cannot exclude the possibility that  infinite-mean models fit   these datasets better than finite-mean models.}  
%  As a standard practice of extreme value theory, a generalized Pareto distribution is usually fitted to the tail region of the data, thus losses exceeding a high threshold (our results can still apply in this case; see Proposition \ref{prop:tail} and Section \ref{sec:62}).
 \cite{IJW09} used  standard seismic theory to show that the tail indices $\alpha$ of earthquake losses lie in the range $[0.6,1.5]$. Estimated by \cite{rizzo2009new}, the tail indices $\alpha$ for some wind catastrophic losses are around $0.7$. \cite{hofert2012statistical} showed that the tail indices $\alpha$ of losses caused by nuclear power accidents are around  $[0.6,0.7]$; similar observations can be found in \cite{sornette2013exploring}. Based on data collected by the Basel Committee on Banking Supervision, \cite{moscadelli2004modelling} reported the tail indices $\alpha$ of (over $40000$) operational losses in $8$ different business lines to lie  in the range $[0.7,1.2]$, with $6$ out of the $8$ tail indices being less than $1$, with $2$ out of these $6$ significantly less than $1$ at  a $95\%$ confidence level. For a detailed discussion on the risk management consequences in this case, see \cite{NEC06}. Losses from cyber risk have tail indices $\alpha\in[0.6,0.7]$; see  \cite{EW19}, \cite{ES20} and  the references therein. In a standard Swiss Solvency Test document (\citet[p.~110]{FINMA2021}), 
most  major damage insurance losses are modelled by a  Pareto distribution with default parameter $\alpha $ in the range $[1,2]$, with  $\alpha=1$ attained by some aircraft insurance.  As discussed by \cite{beirlant1999tail}, some fire losses collected by the reinsurance broker AON Re Belgium have tail indices $\alpha$ around $1$. \cite{BC14} showed that several large commercial property losses collected from two Lloyd’s syndicates have tail indices $\alpha$ considerably less than $1$.  \cite{silverberg2007size} concluded that the tail indices $\alpha$ are less than $1$ for financial returns from some
technological innovations. The tail part of cost overruns in information technology projects can have tail indices $\alpha\le 1$; see \cite{flyvbjerg2022empirical}. In the model of \cite{cheynel2022fraud}, which considers managers’ fraudulent behavior, detected fraud size behaves like a power law with tail index 1. Besides large financial losses and returns, the number  of deaths in major earthquakes and pandemics modelled by Pareto distributions also has infinite mean; see \cite{clark2013note} and \cite{cirillo2020tail}.  City sizes and firm sizes follow Zipf's law ($\alpha\approx 1$); see \cite{gabaix1999zipf} and \cite{axtell2001zipf}. 
Heavy-tailed to extremely heavy-tailed models also occur in the realm of climate change and environmental economics. Weitzman's Dismal Theorem (see \cite{W09}) discusses the breakdown of standard economic thinking like cost-benefit analysis in this context. This led to an interesting discussion with William Nordhaus, a recipient of the 2018 Nobel Memorial Prize in Economic Sciences; see \cite{nordhaus2009analysis}. 

The above references exemplify the occurrence of infinite mean models. Our perspective on these examples and discussions is that if these models are the result of some careful statistical analyses, then the practicing modeler has to take a step back and carefully reconsider the risk management consequences. Of course, in practice, there are several methods available to avoid such extremely heavy-tailed models, like cutting off the loss distribution model at some specific level or tapering (concatenating a light-tailed distribution far in the tail of the loss distribution). In examples like those referred to above, such corrections often come at the cost of a great variability depending on the methodology used. It is in this context that our results add to the existing literature and modeling practice in cases where power-tail data play an important role.

%\end{remark}

 %\subsection{Notation}

   %We fix some notation. 

\section{Preliminaries}\label{sec:2}

 Throughout, random variables are defined on an atomless probability space $(\Omega,\mathcal F,\p)$. Denote by $\mathbb N$ the set of all positive integers and $\R_+$ the set of non-negative real numbers. For $n\in \N$, let $[n]=\{1,\dots,n\}$.  Denote by $\Delta_n$   the standard simplex, that is, $\Delta_n=\{(\theta_1,\dots,\theta_n)\in [0,1]^n:  \sum_{i=1}^n \theta_i=1\}$. 
For $x,y\in \R$, write $x\wedge y= \min\{ x,y\}$ 
%$x\vee y= \max\{ x,y\}$,
and $x_+=\max\{x,0\}$.  We write $\mathbf X\laweq \mathbf Y$ if $\mathbf X$ and $\mathbf Y$ have the same distribution. We always assume $n\ge 2$. Equalities and inequalities are
interpreted component-wise when applied to vectors. For any random variable $X$, its essential infimum is given by $z_X=\inf\{z\in \R: \p(X>z)>0\}$, and its distribution function is denoted by $F_X$. In this paper, all terms like ``increasing" and ``decreasing" are in the non-strict sense.

 We first introduce the super-Pareto distribution and weak negative association,
 and present the main result of \cite{CEW24}, on which most of our study is based. 
The distribution function of a Pareto random variable with parameters $\theta,\alpha>0$ is 
 \begin{align*}%\label{eq:PD}
  P_{\alpha,\theta}(x) = 1 -\left(\frac{\theta}{x}\right)^{\alpha},~~x\ge \theta.
  \end{align*}
 % For $X\sim P_{\alpha,\theta }$,   $q_X(p)=\theta (1-p)^{-1/\alpha}$ for $p\in [0,1)$. 
 The mean of $  P_{\alpha,\theta}$ is infinite (i.e., $  P_{\alpha,\theta}$ is extremely heavy-tailed) if and only if the tail parameter $\alpha \in (0,1]$. 
 % (i.e., $  P_{\alpha,\theta}$ is extremely heavy-tailed if $\alpha \in (0,1]$). 
 As it suffices to consider $  P_{\alpha,1}$ in this paper, we write   $ P_{\alpha,1} $ as $ \mathrm {Pareto}(\alpha)$.
 % We say that the $  P_{\alpha,\theta}$ distribution is \emph{extremely heavy-tailed} if $\alpha\le 1$. 
 %and it is \emph{moderately heavy-tailed} if $\alpha >1$.  

\begin{definition} 
A random variable   $X$ is \emph{super-Pareto} (or has a super-Pareto distribution)  if $X\laweq f(Y)$ for some increasing, 
convex, and non-constant function $f$ and $Y\sim \mathrm{Pareto}(1)$. 
% Moreover, $X$ is \emph{regular} if 
% $f(0)=0$ and $f(1)>0$.  
\end{definition}
As super-Pareto losses can be obtained by increasing and convex transforms of Pareto($1$) losses, their tails are heavier than (or equivalent to) those of Pareto($1$) losses and thus have infinite mean. All extremely heavy-tailed Pareto distributions are super-Pareto. {Other examples of super-Pareto distributions include generalized Pareto distributions and  Burr distributions, of which paralogistic distributions and log-logistic distributions are special cases; see \cite{CEW24}.}

\begin{definition}
A set $S\subseteq \R^{k}$, $k\in \N$, is  \emph{decreasing} 
if $\mathbf x\in S$ implies $\mathbf y\in S$ for all $\mathbf y\le \mathbf x$. Random variables $X_1,\dots,X_n$ are \emph{weakly negatively associated} if  
 for any $i\in[n]$,  decreasing set $S  \subseteq  \R^{n-1}$, and $x\in \R$ with $\p(X_i\le x)>0$,   it holds that 
 $
 \p(\mathbf X_{-i} \in S  \mid  X_i\le  x) \le \p(\mathbf X_{-i}\in S), $
 where $\mathbf X_{-i}=(X_1,\dots,X_{i-1}, X_{i+1},\dots,X_n)$.
\end{definition}
% Intuitively, \eqref{eq:WNA} means that if $X_i$ is small, $\mathbf X_{-i}$ is less likely to be small.
Weak negative association includes independence as a special case. Weak negative association  is implied by negative association (\cite{AS81}  and \cite{JP83}) and negative regression dependence (\cite{L66} and \cite{block1985concept}), and it implies negative orthant dependence (\cite{block1982some}). Multivaraite normal distributions with non-positive correlations are negatively associated and thus are weakly negatively associated.

This paper mainly considers weakly negatively associated and identically distributed (WNAID) super-Pareto random variables. This setting includes iid Pareto losses with infinite mean.
For random variables $X$ and $Y$, we write $X<_{\rm st}Y$ if $\p(X>x)<\p(Y>x)$ for all $x\in \R$ satisfying $\p(X>x)\in (0,1)$, and this will be referred to as strict stochastic dominance.\footnote{This condition is stronger than a different notion of strict stochastic dominance defined by $X\le _{\rm st} Y$ and $X\not \ge_{\rm st} Y$.}  
The following result serves as the basis for our study. 
\begin{theorem}[\cite{CEW24}]\label{thm:1}
Let
 $X_1,\dots,X_n$ be WNAID super-Pareto and $X\laweq X_1$. 
For $(\theta_1,\dots,\theta_n)\in\Delta_n$, we have
%  \begin{enumerate}[(i)]
% \item  
%Stochastic dominance holds:
\begin{equation} \label{eq:maineq1} 
X \le_{\rm st}\sum_{i=1}^n\theta_{i}X_{i} .
\end{equation}
Moreover, strict stochastic dominance $X <_{\rm st}\sum_{i=1}^n\theta_{i}X_{i}$ holds if $\theta_i>0$ for at least two $i\in [n]$.
% \item  If $X$ is regular, then for any events  $A_1,\dots,A_n$  independent of $(X_1,\dots,X_n)$ and event $A$ independent of $X$  satisfying $\p(A)=\sum_{i=1}^n \theta_i\p(A_i)$,  we have
% \begin{equation} \label{eq:maineq2} 
% X \id_A\le_{\rm st}\sum_{i=1}^n\theta_{i}X_{i}\id_{A_i}.
% \end{equation}
% \end{enumerate}
\end{theorem}

Inequality \eqref{eq:maineq1} implies that for an agent who wants to minimize the default probability of their loss portfolio, the optimal strategy is non-diversification.  
A similar inequality to \eqref{eq:maineq1} also holds in a model for catastrophic losses (i.e., a loss can be written as $Y\id_{A}$ where $Y$ is a loss given the occurrence of a catastrophic event $A$), obtained in Theorem 1 (ii) of \cite{CEW24}, which we omit.
Other generalizations of \eqref{eq:maineq1} will be obtained in Section \ref{sec:extension}.
%This also applies to almost all useful decision models under risk; see Section \ref{sec:4}. 
 
\begin{remark}
In most part of the paper, $X_1,\dots,X_n$  are assumed to be super-Pareto risks in \eqref{eq:maineq1} as well as in other results built upon \eqref{eq:maineq1}, such as generalizations of Theorem \ref{thm:1} in Section \ref{sec:extension} and risk management implications of \eqref{eq:maineq1} in Sections \ref{sec:4} and \ref{app:r1}. We note that
\eqref{eq:maineq1} and the associated results hold for  some  other models, not covered in this paper. 
After the first version of this paper, there have been some generalizations on the distributions for \eqref{eq:maineq1} to hold, such as iid super-Cauchy random variables in \cite{M24}, iid InvSub random variables in \cite{ALO24}, and negatively lower orthant dependent (\cite{block1982some}) and identically distributed super-Fr{\'e}chet random variables in  \cite{CS24}.
    Moreover, the earlier work of \cite{I05} also contains the class of iid positive one-sided stable random variables for \eqref{eq:maineq1}.  All of these models have infinite or undefined mean.  Our results can be easily adapted to these models.
    For a review on recent results on infinite-mean models in risk management, see \cite{CW25}.
\end{remark} 
\begin{remark}\label{rem:EVT}
A distribution $F$ is said to be  \emph{regularly varying} with tail parameter $\alpha\ge0$, if
$\overline F(x)= L(x)x^{-\alpha}$
where $L$ is a \emph{slowly varying} function, that is, $L(tx)/L(x)\rightarrow 1$ as $x\rightarrow \infty$ for all $t>0$; see \cite{EKM97}. Pareto distributions are regularly varying. If $X_1,\dots,X_n$ are iid and follow a regularly varying distribution with tail parameter less than 1 (i.e., they have infinite mean), by Lemma 1.3.1 of \cite{EKM97}, 
\begin{equation}\label{eq:EVT}
\lim_{t\rightarrow \infty}\frac{\p\(\sum_{i=1}^nX_i/n>t\)}{\p(X_1>t)}\ge 1,
\end{equation}
which can be regarded as an asymptotic version of Theorem \ref{thm:1}. Our main relation \eqref{eq:maineq1}  does not hold for regularly varying distributions in general, 
because it is a property of the entire distribution, not only its asymptotic behaviour.     For instance, we can modify the middle part of the distribution of $X$ to break \eqref{eq:maineq1}, yet preserving regular variation. 
\end{remark}

 \section{Generalizations of the super-Pareto stochastic dominance}\label{sec:extension}

We discuss whether and how \eqref{eq:maineq1} can be generalized beyond the WNAID super-Pareto model.
% Given the risk management implications of  \eqref{eq:maineq1}, 
% one naturally wonders whether and how it can be generalized beyond the WNAID super-Pareto model.
That is, whether  
\begin{align}
\label{eq:R2-new}
X\le_{\rm st} \sum_{i=1}^n \theta_i X_i ~\mbox{for all $(\theta_1,\dots,\theta_n)\in \Delta_n$, where $X,X_1,\dots,X_n$ are identically distributed}\tag{DP}
\end{align} 
holds under   models other than those covered in Theorem \ref{thm:1} (``DP" stands for ``diversification penalty").
We consider two directions of generalizations, one on the marginal distributions and one on the dependence structure.  Below, $n\ge 2$ is fixed. 
First, let
$$
\mathcal F_{\mathrm{IN}}=\{\mbox{distribution of $X$}:\mbox{\eqref{eq:R2-new} holds for all independent  $X_1,\dots,X_n$}\}.
$$
The set $\mathcal F_{\mathrm{WNA}}$ is defined similarly, where the subscript WNA means weak negative association instead of independence, among $X_1,\dots,X_n$. Clearly, $\mathcal F_{\mathrm{WNA}}\subseteq \mathcal F_{\mathrm{IN}}$, and   each of them contains all super-Pareto distributions by Theorem \ref{thm:1}.
Below, by considering transforms on $\mathcal F_{\mathrm{WNA}}$ and $ \mathcal F_{\mathrm{IN}}$, we mean that the transform is applied to the random variable $X$. 
% The sets $\mathcal F_{\mathrm{WNA}}$ and $\mathcal F_{\mathrm{NA}}$ are defined similarly, where the subscript (W)NA means  (weak) negative association, instead of independence, among $X_1,\dots,X_n$. Clearly, $\mathcal F_{\mathrm{WNA}}\subseteq \mathcal F_{\mathrm{NA}}\subseteq \mathcal F_{\mathrm{IN}}$, and   each of them contains all super-Pareto distributions by Theorem \ref{thm:1}.
% Below, by considering transforms on $\mathcal F_{\mathrm{WNA}}$, $\mathcal F_{\mathrm{NA}}$, and $ \mathcal F_{\mathrm{IN}}$, we mean that the transform is applied to the random variable $X$. 

\begin{proposition}\label{prop:convolution}
The set  $ \mathcal F_{\mathrm{IN}}$ is closed under convolution.
Both $\mathcal F_{\mathrm{WNA}}$ and $ \mathcal F_{\mathrm{IN}}$ are  closed under strictly increasing convex transforms. 
\end{proposition}

The second statement in Proposition \ref{prop:convolution} means \eqref{eq:R2-new} also holds for strictly increasing convex transforms of $X,X_1,\dots,X_n$ given some dependence assumptions (i.e., $f(X)\le_{\rm st} \sum_{i=1}^n \theta_i f(X_i)$ where $f$ is any strictly increasing convex function).

Second, we consider possible dependence structures for \eqref{eq:R2-new}. Copulas are useful tools for modeling dependence structures; see  \cite{N06} for an overview. 
A \emph{copula} is a distribution function with standard uniform (i.e., on $[0,1]$) marginal distributions. 
For a random vector $\mathbf X$   with distribution function $F$ and marginal distributions $F_1,\dots,F_n$, by Sklar's Theorem (e.g., Theorem 7.3 of \cite{MFE15}), there exists a copula  $C$ satisfying  $F(x_1,\dots,x_n)=C(F_1(x_1),\dots,F_n(x_n))$ for $(x_1,\dots,x_n)\in \R^n$, and such $C$ is called a copula of $\mathbf X$, which is unique when $\mathbf X$ has continuous marginal distributions.
The \emph{copula for independence} and the \emph{copula for comonotonicity} 
are given by $ \mathbf u\mapsto  \prod_{i=1}^nu_i$  and  $ \mathbf u \mapsto \min(\mathbf u)$ for $\mathbf u=(u_1,\dots, u_n)\in[0,1]^n$, respectively.
A \emph{copula for weak negative association} is any copula of weakly negatively associated standard uniform random variables. 
As their names suggest, these  copulas represent the corresponding dependence structures.
The set of dependence structures satisfying \eqref{eq:R2-new} is then represented by
$$
\mathcal C_{\mathrm{DP}}=\{\mbox{copula } C:\mbox{\eqref{eq:R2-new} holds for all super-Pareto  $X_1,\dots,X_n$ with copula $C$}\}.
$$
%For $(u_1,\dots u_n)\in[0,1]^n$, the independence copula and the comonotonicity copula are $(u_1,\dots u_2)\mapsto \prod_{i=1}^nu_i$  and  $(u_1,\dots u_2)\mapsto \min_{i\in[n]}u_i$. 
% Let $X, X_1,\dots,X_n$ be identically distributed super-Pareto random variables, $\mathcal C_n$ be the set of all $n$-dimensional  copulas, and $C\in \mathcal C_n$. We say  $X \le_{\rm st}\sum_{i=1}^n\theta_{i}X_{i}$ holds for $C$ if it holds when $C$ is the copula of $X_1,\dots,X_n$. 
\begin{proposition}\label{prop:mixture}
The set $\mathcal C_{\mathrm{DP}}$ is closed under mixture, and it contains all copulas for comonotonicity, independence, and weak negative association.
% Let $\mathcal C\subseteq \mathcal C_n$ and $X, X_1,\dots,X_n$ be identically distributed super-Pareto random variables. For $(\theta_1,\dots,\theta_n)\in\Delta_n$, if 
% $$X \le_{\rm st}\sum_{i=1}^n\theta_{i}X_{i},$$
% holds for each copula in $\mathcal C$, then it holds for any copula in the convex hull of $\mathcal C$. 
\end{proposition}

%\begin{proof}[Proof of Proposition \ref{prop:mixture}]
%Denote by $C$ the copula of  $(X_1,\dots,X_n)$. Then,  there exists a random vector $(U_1,\dots,U_n)\sim C$ such that $(q_{X_1}(U_1),\dots,q_{X_n}(U_n))$ has the same distribution as $(X_1,\dots,X_n)$. Note that for $p\in(0,1)$, $\p(\sum_{i=1}^n\theta_{i}X_{i}\le p)=\p(\sum_{i=1}^n\theta_{i}q_{X_i}(U_i)\le p) $  is linear in the distribution of $(U_1,\dots,U_n)$.  Therefore, if   $\p(\sum_{i=1}^n\theta_{i}X_{i}\le p)\le \p(X_1\le p) $ for all $p\in(0,1)$ holds for every element in $\mathcal C$, it also holds for  every element from the convex hull of $\mathcal C$.
%\end{proof}

% By Lemma \ref{thm:1}, if identically distributed super-Pareto random variables $X_1,\dots,X_n$ are independent or weakly negatively dependent, \eqref{eq:maineq1} holds. It is  clear that \eqref{eq:maineq1} also holds if $X_1,\dots,X_n$ are perfectly positively dependent. Hence,  
By Proposition $\ref{prop:mixture}$, \eqref{eq:R2-new} holds under some particular forms of positive or mixed dependence, in addition to the weak negative association studied by \cite{CEW24}. Nevertheless, we did not find a natural model of positive dependence, other than a mixture of independence and comonotonicity, that yields \eqref{eq:R2-new}.

%The inequality \eqref{eq:maineq} also holds for some {positively} correlated extremely heavy-tailed Pareto losses.  First, the inequality \eqref{eq:maineq} simply holds for perfectly positively dependent extremely heavy-tailed Pareto losses $Y_1,\dots,Y_n$ (i.e., $Y_1=\dots=Y_n$ almost surely).
%Note that there exists a function $f$ such that \eqref{eq:maineq} can be rewritten as $\p(f(U_1,\dots,U_n) \le p)\ge p$ for all $p\in[0,1]$ where $(U_1,\dots,U_n)$ is a vector of standard uniform random variables (see the proof of Lemma \ref{thm:1} in Section \ref{app:proof}). Since $\p(f(U_1,\dots,U_n\le p) $ is linear in the distribution of $(U_1,\dots,U_n)$ (i.e., copula),  \eqref{eq:maineq} remains true if the dependence structure  of risks $Y_1,\dots,Y_n$ is a mixture of independence and perfectly positive dependence, that is, $(U_1,\dots,U_n)$ follows the following copula  

We present below two additional models for which similar results to Theorem \ref{thm:1} hold: the tail super-Pareto distribution model,
and the collective risk model in insurance.
% Generalizations of \eqref{eq:maineq1} to a tail super-Pareto risk model and a classic insurance risk model are presented below.
 
 %\subsection*{Tail Pareto distributions}\label{app:gen}

 %\subsection{Tail super-Pareto distributions}

As reflected by the Pickands-Balkema-de Haan Theorem (see Theorem 3.4.13 (b) of \cite{EKM97}), 
many losses have a power-like tail, but their distributions may not be power-like over the full support.
Therefore, it is practically useful to assume that a random loss has a Pareto distribution only in the tail region; see the examples in the Introduction.
%For $\alpha>0$, 
%we say that $Y$ has a {Pareto}($\alpha$) distribution beyond $x\ge 1$  if  $\p(Y>t)=t^{-\alpha}$ for $t\ge x$. 

Let $X$ be a super-Pareto random variable. We say that $Y$ is distributed as $X$  beyond $x $  if  $\p(Y>t)=\p(X>t)$ for $t\ge x$. 
Our next result suggests that, under an extra condition,  a stochastic dominance also holds in the tail region for such distributions.

\begin{proposition}
\label{prop:tail}
Let $X$ be a super-Pareto random variable,    $Y_{1},\dots,Y_{n}$ be iid random variables distributed as  $X$ beyond $x\ge z_X$, and $Y\laweq Y_1$. Assume that $Y\ge_{\rm st} X$.
  For $(\theta_1,\dots,\theta_n)\in\Delta_n$ and $t\ge x$, we have  $\p\(\sum_{i=1}^n\theta_{i}Y_{i}>  t\)\ge  \p\(Y> t\)$, and the inequality is strict if $t>z_X$ and $\theta_i>0$ for at least two $i\in [n]$.
\end{proposition}

In Proposition \ref{prop:tail}, the assumption $Y\ge_{\rm st} X$, that is,
$\p(Y>t) \ge \p(X>t)$ for $t\in [z_X,x]$,  is not dispensable.
Here we cannot allow the distribution of $Y$ on $[z_X,x]$ to be arbitrary; the entire distribution is relevant to establish the inequality $\p\(\sum_{i=1}^n\theta_{i}Y_{i}>  t\)\ge  \p\(Y> t\)$, even for $t$ in the tail region.

%Let $ X,X_{1},\dots,X_{n}$ be iid $\mathrm {Pareto}(\alpha)$ random variables with $\alpha\in(0,1]$. 
%As a particular application of Proposition \ref{prop:tail}, it holds that, for any $m\ge 1$,
% \begin{align}
% \label{eq:truncate-max}
%X\vee m \le_{\rm st}\sum_{i=1}^n\theta_{i}(X_{i}\vee m).
%\end{align}
%This inequality follows by noting that $X\vee m$ has a Pareto distribution beyond $m$ and applying Proposition \ref{prop:tail} to $t\ge m$.  A location shift  of \eqref{eq:truncate-max} also gives 
% \begin{align}
% \label{eq:truncate-excess} 
%(X - m)_+ \le_{\rm st}\sum_{i=1}^n\theta_{i}(X_{i}-m)_+.
%\end{align}
%For \eqref{eq:truncate-max} and \eqref{eq:truncate-excess} to hold, it suffices to assume that $X_{1},\dots,X_{n}$ are $\mathrm {Pareto}(\alpha)$ beyond $m$, as their distribution on $(-\infty,m]$ does not matter. 

% By applying Proposition \ref{prop:tail}, we get 
%$\p(\sum_{i=1}^n\theta_{i}(X_{i}\vee m)>t)\ge 
%\p( X\vee m>t)$ for $t\ge m$.
%For $t\le m$, we have $\p(\sum_{i=1}^n\theta_{i}(X_{i}\vee m)>t)=1= 
%\p( X\vee m>t)$.

 % \subsection*{A classic model in insurance}\label{app:insurance}
 
Random weights and a random number of risks are, for instance, common in modeling portfolios of insurance losses; see \cite{KPW12}.
 Let $N$ be a counting random variable (i.e., it takes values in $\{0,1,2,\dots\}$), and $W_i$ and $X_i$ be positive random variables for $i\in \N$.
We consider an insurance portfolio where each policy incurs a loss $W_iX_i$ if there is a claim, and $N$ is the total number of claims in a given period.  If $W_1=W_2=\dots=1$ and $X_1,X_2,\dots$ are iid, then this model recovers the classic collective risk model.
The portfolio loss of insurance policies  is given by
${\sum_{i=1}^NW_iX_i} ,$
and its average loss across claims is 
$({\sum_{i=1}^NW_iX_i})/({\sum_{i=1}^N W_i})$
where both terms are $0$ if $N=0$.
\begin{proposition}\label{cor:random}
Let $X_1,X_2,\dots$ be WNAID super-Pareto random variables, $X\laweq X_1$, $W_1,W_2,\dots$ be positive random variables, and $N$ be a counting random variable, such that $X,\{X_{i}\}_{i \in \N}$, $\{W_i\}_{i\in\N}$, and $N$ are independent. We have 
\begin{equation}
\label{eq:collective} X \id_{\{N\ge1\}}\le_{\rm st}\frac{\sum_{i=1}^NW_iX_i}{\sum_{i=1}^N W_i} \mbox{~~~and~~~}\sum_{i=1}^N W_i X  \le_{\rm st} {\sum_{i=1}^NW_iX_i}.\end{equation}
  If $\p(N\ge2)\neq0$, then for $t>z_X$,   $\p (\sum_{i=1}^NW_iX_i/\sum_{i=1}^NW_i\le t )<\p (X\id_{\{N\ge 1\}}\le t )$. 
\end{proposition}

 If $W_1=W_2 =\dots=1$ as in the classic collective risk model, then, under the assumptions of Proposition \ref{cor:random}, we have 
$$X_1\id_{\{N\ge1\}}\le_{\rm st} \frac 1N  {\sum_{i=1}^N X_i}  \mbox{~~~and~~~}N X_1 \le_{\rm st} {\sum_{i=1}^NX_i}.$$

The above inequalities suggest that the sum of a randomly counted sequence of  WNAID super-Pareto losses is stochastically larger than the sum of a randomly counted sequence of identical super-Pareto losses. Therefore, building an insurance portfolio for WNAID super-Pareto claims does not reduce the total risk.
In this setting, it is less risky to insure identical policies than to insure weakly negatively associated policies of the same type of super-Pareto loss and thus the basic principle of insurance does not apply to super-Pareto losses.

\section{Risk management implications}\label{sec:4}
%\subsection{No diversification for a monotone agent}
\subsection{Decision models for infinite-mean risks}

The interpretation of Theorem \ref{thm:1} and its several generalizations in Section \ref{sec:extension} is that non-diversification is better than diversification in a pool of extremely heavy-tailed losses. We first discuss useful decision models for which this result can or cannot be applied.  

Denote by $\mathcal X$ the set of all random variables and  $L^1\subseteq \mathcal X$  the set of random variables with finite mean.
Assume that $X_1,\dots,X_n$ are WNAID super-Pareto and $(\theta_1,\dots,\theta_n)\in\Delta_n$ with at least two of $\theta_1,\dots,\theta_n$ being positive. 
 Let $\mathcal L\subseteq \mathcal X$ be a
set of random variables representing possible losses to an agent, such that $\sum_{i=1}^n\theta_{i}X_{i}\in \mathcal L$ and $X_1\in \mathcal L$.
Moreover, assume that $\mathcal L$ is law invariant; that is, $X\in\mathcal L$ and $X\laweq Y$ imply $Y\in \mathcal L$.

The preferences of the agent are represented by a transitive binary relation
$
\succeq
$  (with strict relation $\succ$ and symmetric relation $\simeq$) on $\mathcal L$,
which we assume to satisfy two natural properties:
\begin{enumerate}[(i)]
\item Choice under risk: $X\simeq Y$ for $X,Y\in\mathcal L$ if $X\laweq Y$; 
\item Less loss is better: $X\succeq Y$ for $X,Y\in\mathcal L$ if $X\le Y$ (in the almost sure sense, omitted below).
\end{enumerate}
  By Theorem \ref{thm:1} and the representation of first-order stochastic dominance (see e.g., \citet[Theorem 1.A.1]{SS07}), 
$$
X_1  \le_{\rm st}  \sum_{i=1}^n\theta_{i}X_{i} \implies \mbox{there exists $Y\laweq X_1$ such that } Y  \le   \sum_{i=1}^n\theta_{i}X_{i}  \implies X_1   \simeq Y \succeq   \sum_{i=1}^n\theta_{i}X_{i}.
$$
A standard construction of $Y$ is to let $Y=h(U)$, where $U$ is uniformly distributed on $[0,1]$ such that $\sum_{i=1}^n\theta_{i}X_{i} =g(U)$, and $g$ and $h$ are the left quantile functions of $\sum_{i=1}^n\theta_{i}X_{i} $ and $X_1$, respectively. 
Therefore, any agent  satisfying (i) and (ii) would (weakly) prefer non-diversification modelled by  $X_1$ to diversification modelled by $ \sum_{i=1}^n\theta_{i}X_{i}$.  
Moreover, we can further consider 
\begin{enumerate}[(i)] 
\item[(iii)]  Less loss is strictly better: $X\succ Y$ for $X,Y\in\mathcal L$ if $\p(X<Y)=1$.
\end{enumerate}
If (iii) holds, then we have a strict preference 
$ 
X_1  \succ  \sum_{i=1}^n\theta_{i}X_{i},
$ following the argument above.

%Next, we use a functional $\rho$ to represent the binary relation. Measuring the risk of a financial portfolio is a crucial task in both the banking and insurance sectors, and it is typically done by calculating the value of a risk measure.
Risk measures are an important tool to evaluate portfolio risks for financial institutions and many of them induce the binary relation discussed above.
 A \emph{risk measure} is a functional $\rho: \mathcal X_{\rho}\rightarrow \overline \R:= [-\infty,\infty]$, where the domain $\mathcal X_\rho\subseteq \mathcal X$ is a set of random variables representing financial losses. We will assume that an agent uses a risk measure $\rho$ for their preference, in the sense that 
 $\rho(X)\le \rho(Y) \Leftrightarrow X\succeq Y$. 
Our notion of a risk measure is quite broad, and it includes not only risk measures in the sense of \cite{ADEH99} and \cite{FS16} but also 
decision models such as the expected utility by flipping the sign. 
The   assumptions on $\rho$ below correspond to (i), (ii) and (iii) respectively.
 \begin{enumerate}[(a)]
 \item Law invariance: $\rho(X) = \rho(Y)$ for $X,Y\in\mathcal X_\rho$ if $X\laweq Y$.
 \item Weak monotonicity: $\rho(X) \le \rho(Y)$ for $X,Y\in\mathcal X_\rho$ if $X\le_{\rm st} Y$. 
 \item Mild monotonicity: $\rho$ is weakly monotone and 
$\rho(X)< \rho(Y)$ for $X,Y\in\mathcal X_\rho$ if $\p(X<Y)=1$.
\end{enumerate} 
 
% A binary relation with properties (i) and (ii) can be induced by 
% a risk measure satisfying (a) and (b). If the binary relation further satisfies (iii), then it can be induced by 
% a risk measure satisfying (a) and (c).
% It is well-known that for  any law-invariant functional $\rho:\X_\rho\to \R$, weak monotonicity is  equivalent to the   standard monotonicity (e.g., Theorem 6.28 in \cite{shapiro2021lectures}): $\rho(X)\le \rho(Y)$ for $X,Y\in\mathcal X_\rho$ if $X\le Y$.
% %\item $\rho$ is weakly monotone. 
% %\item $\rho$ is increasing (in the sense that $\rho(X)\le \rho(Y)$ if $X\le Y$) and law-invariant (in the sense that $\rho(X)=\rho(Y)$ if $X\laweq Y$);
% %\item $\rho$ is weakly monotone (in the sense that $\rho(X) \le \rho(Y)$  if $X\le_{\rm st} Y$).  
% This property is satisfied by all commonly used risk measures, as well as expected utilities and hence they are weakly monotone.
% As such, stochastic dominance is a natural notion to be used for decision making.

Many commonly used decision models are not just weakly monotone but mildly monotone; we highlight some examples.
First, for an increasing utility function $u$, the expected utility agent can  be represented by  a risk measure $E_v$, namely
$$
E_v(X) = \E[v(X)],~~~~~~X\in \X_{E_v}:=\{Y\in \X: \E[|v(Y)|]<\infty\},
$$
where $v(x)=-u(-x)$ is also increasing.
It is clear that $E_v$ 
is   mildly monotone if $v$  or $u$ is strictly increasing. Note that
for risk-averse expected utility decision makers ($u$ is concave), the domain of $E_v$ is typically smaller than $\mathcal L$. 
However, there are still a few useful contexts where expected utility models can be used
to compare infinite-mean losses. 
Below we give a few examples. 
\begin{enumerate}[(a)]
\item Suppose that for some $\ell \in \R$, $u(x)=u(\ell)$ for all $x\le \ell$, and $u$ is concave on $(\ell,\infty)$. This utility function describes a risk-averse agent with limited liability.
Limited liability is a practical assumption in banking and insurance decisions for both individuals and financial institutions. 
\item The Markowitz utility function (\cite{M52}) has a convex-concave structure on the loss side, which is based on  the empirical observation  that people often prefer   a loss of $10m$ dollars with $0.1$ probability over a sure loss of $m$ dollars when $m$ is very large.
\item In the cumulative prospect theory (which generalizes the expected utility model) of \cite{TK92}, the utility function 
is convex below a reference point. 
\end{enumerate} 
In all cases above, the expected utility can take finite values on $\mathcal L$ and is mildly monotone, implying a strict preference for non-diversification.

The next examples of mildly monotone risk measures are the two widely used regulatory risk measures in insurance and finance, Value-at-Risk (VaR) and Expected Shortfall (ES).
  For $X\in \mathcal{X}$ and $p\in (0,1)$, VaR is defined as
  \begin{equation*}%\label{VaR}
  \VaR_{p}(X)=F^{-1}_X(p)=\inf\{t\in\R:F_X(t)\geq p\},
  \end{equation*}
and ES is defined as
$$\ES_{p}(X)=\frac{1}{1-p}\int_{p}^{1}\VaR_{u}(X)\mathrm{d}u.$$
%If $F_X$ is continuous, then $\ES_{p}(X)=\E[X|X\ge \VaR_p(X)]$. 
 For $X\notin L^1$, such as super-Pareto losses, $\ES_p(X)$ can be $\infty$, whereas $\VaR_p(X)$ is always finite. Therefore,
VaR is mildly monotone on $\X$, whereas ES is mildly monotone only on $L^1$. Note that convex risk measures will take infinite values  when evaluating infinite-mean losses and hence are not suitable for losses in $\mathcal L$; standard properties of risk measures are collected in Appendix \ref{app:A}.

\subsection{Diversification penalty for losses with bounded support}
Although  \eqref{eq:maineq1}  never holds for non-degenerate random variables with finite mean (Proposition 2 in \cite{CEW24}), Theorem \ref{thm:1} can be applied to some contexts of finite-mean losses. 
 Since (strict) first-order stochastic dominance is preserved under (strictly) increasing transformations, 
Theorem \ref{thm:1} implies 
 $
 h(X_1) \le_{\rm st} h(\sum_{i=1}^n \theta_iX_i)
 $
 for all increasing real functions $h$.
 For instance, in the context of reinsurance, $h(x)=x\wedge c$ for some threshold $c\in \R$ corresponds to an excess-of-loss  reinsurance coverage; see e.g., \cite{O18}. 
 On the other hand, for an increasing transform $h$ performed on  individual losses, the first-order stochastic dominance 
  $
 h(X_1) \le_{\rm st}\sum_{i=1}^n  \theta_i h( X_i)
 $
 may not hold, unless $h$ is convex (see Lemma 2  in \cite{CEW24}).
 Especially, if $\E[h(X_1)]<\infty$ (this cannot happen if $h$ is convex and non-constant) and $X_1,\dots,X_n$ are iid, then 
  $h(X_{1})\ge_{\rm cx}\sum_{i=1}^n  \theta_i h( X_i)$ holds,  where $Y\ge_{\rm cx} Z$ means $\E[\phi(Y)] \ge \E[\phi(Z)]$ for all convex $\phi$ such that the expectations exist. 
With finite mean, a risk-averse expected utility agent would favour   diversification, a well-known phenomenon; see, e.g., \cite{S67}. 
Although   $
 g(X_1) \le_{\rm st}\sum_{i=1}^n  \theta_i g( X_i)
 $ fails to hold in case $g$ is bounded,  non-diversification is still preferred for $\VaR_p$ with specific $p$ when $g(x)=x\wedge c$. We present a generalization, allowing a different threshold for each loss, in the following result.

\begin{theorem}\label{prop:bounded}
   Let $X,X_{1},\dots,X_{n}$ be WNAID super-Pareto random variables and $Y_i= X_i\wedge c_i$ where $c_i\ge z_{X}$ for each $i\in[n]$. Suppose that   $(\theta_1,\dots,\theta_n)\in\Delta_n$ such that $\theta_i>0$ for at least two $i\in [n]$, and denote by  $c=\min\{c_1\theta_1,\dots,c_n\theta_n\}$. We have 
\begin{equation*}
\p\(\sum_{i=1}^n\theta_{i}Y_{i}>  t\)=\p\(\sum_{i=1}^n\theta_{i}X_{i}>  t\)>\p\(X> t\)= \p\(Y_i> t\), ~~~i\in[n]\end{equation*}
for $t\in(z_{X},c]$,  and  
$$   \VaR_{p}\left( \sum_{i=1}^n \theta_{i}Y_{i}\right) >  \sum_{i=1}^n\theta_i\VaR_{p}(Y_i)$$  for $p\in(0,  \P(X\le c) ) $.
\end{theorem} 
 %
% 
% As $Y_{1},\dots,Y_{n}$ are bounded and integrable, for $p\in (0,1)$ and $(\theta_1,\dots,\theta_n)\in\Delta_n$,  $\ES_{p}\(\theta_{1}Y_{1}+\dots+\theta_{n}Y_{n}\)\le\theta_{1}\ES_{p}(Y_{1})+\dots+\theta_{n}\ES_{p}(X_{n})$. However, $\ES_{p}\(\theta_{1}Y_{1}+\dots+\theta_{n}Y_{n}\)$ can be too large to be held by financial institutions. Thus, VaR and RVaR still play an important role even if Pareto losses are bounded.   For $\rho=\VaR$ or $\RVaR$,  the convexity property may or may not hold,  depending on the choice of threshold. 
%We will see that diversification penalty still exists in this model.

Theorem \ref{prop:bounded} states that for a fixed weight vector of positive components and a fixed  $p\in (0,1)$, if the  thresholds $c_1,\dots,c_n$ are high enough, % implying that $c$ is large,
then non-diversification is better than diversification.
% when the risk is measured by $\VaR_p$.
A closely related observation was made by \cite{IW07}:
For iid symmetric infinite-mean stable random variables truncated by a sufficiently high threshold, diversification makes the portfolio “more spread out” and thus more risky.
% $\p\(\sum_{i=1}^n\theta_{i}Y_{i}>  t\)$ coincides with $\p\(\sum_{i=1}^n\theta_{i}X_{i}>  t\)$, and \eqref{eq:result-VaR-2}  holds for $p\in(0,1)$ not too large. 
 %The second question is whether 
%\begin{align}\label{eq:q2}
%\VaR_{p}\( \sum_{i=1}^n \theta_{i}X_{i}\)\ge \sum_{i=1}^n \theta_{i}\VaR_{p}(X_{i})
%\end{align}
%holds for $(\theta_1,\dots,\theta_n)\in\Delta_n$ and independent extremely heavy-tailed Pareto losses $X_1,\dots,X_n$ with possibly different tail parameters.
% It remains unclear whether \eqref{eq:result-VaR-2} or its non-strict version holds for iid super-Pareto losses $X_1,\dots,X_n$ with possibly different distributions.

 \subsection{No diversification for a single agent} 
%\label{app:single-agent}

Next, we formalize the decisions of a single agent in a risk sharing pool. 
 Suppose that $Y_1,\dots,Y_n$ are WNAID super-Pareto and $Y\laweq Y_1$.
From now on, we will assume that  $\mathcal X_\rho$ contains the random variables in $\mathcal L$.
% ; this puts some restrictions on $v$ for $E_v$ since  super-Pareto random variables do not have finite mean. 
The following result on the diversification penalty of super-Pareto losses for a monotone agent follows directly from Theorem \ref{thm:1}. 
% Next we formalize the decisions of a single agent in a   risk sharing pool. 
%  We consider model A, where $Y_1,\dots,Y_n$ are WNAID super-Pareto and $Y\laweq Y_1$,  and model B, where $Y=X\id_A $ and $Y_i=X_i\id_{A_i}$ for $i\in [n]$ in Theorem \ref{thm:1} (ii).
% From now on, we will assume that  $\mathcal X_\rho$ contains the random variables in models A and B; this puts some restrictions on $v$ for $E_v$ since  super-Pareto random variables do not have finite mean. 
% The following result on the diversification penalty of super-Pareto losses for a monotone agent follows directly from Theorem \ref{thm:1}. 

\begin{proposition}\label{prop:risk-meas}
 For $(\theta_1,\dots,\theta_n)\in\Delta_n$ and  a   weakly monotone risk measure $\rho:\mathcal X_\rho\rightarrow \overline \R$,  we have
\begin{align}\label{no convexity}
\rho\left( \sum_{i=1}^n \theta_{i}Y_{i}\right)\ge \rho(Y).
\end{align}
 The inequality in \eqref{no convexity} is strict if $\rho$ is mildly monotone and $\theta_i>0$ for at least two $i\in [n]$.
\end{proposition}

We distinguish strict and non-strict inequalities in \eqref{no convexity} because a strict inequality has stronger implications on the optimal decision of an agent.   As an important consequence of Proposition \ref{prop:risk-meas}, for $p\in (0,1)$ and $(\theta_1,\dots,\theta_n)\in \Delta^n$, 
\begin{equation}\label{eq:result-VaR}\VaR_{p}\left( \sum_{i=1}^n \theta_{i}Y_{i}\right)\ge \VaR_p(Y),\end{equation} 
and the inequality is strict if $\theta_i>0$ for at least two $i\in [n]$. 
% \begin{equation}\label{eq:result-VaR-2}\VaR_{p}\left( \sum_{i=1}^n \theta_{i}Y_{i}\right) >  \sum_{i=1}^n\theta_i\VaR_{p}(Y_i).\end{equation}  
Inequality \eqref{eq:result-VaR} and its strict version will be referred to as the diversification penalty for $\VaR_p$. Since all commonly used decision models are mildly monotone,
Proposition \ref{prop:risk-meas} and \eqref{eq:result-VaR} suggest that diversification of super-Pareto losses is detrimental.  
\begin{remark}
Inequality \eqref{eq:result-VaR}, known as \emph{superadditivity} of VaR, is different from the asymptotic superadditivity of VaR in the literature, that is, if $X_1,\dots,X_n$ are iid and regularly varying  with tail parameter less than 1 (see Remark \ref{rem:EVT}),
$$\lim_{p\rightarrow 1}\frac{\VaR_p(X_1+\dots+X_n)}{\VaR_p(X_1)+\dots+\VaR_p(X_n)}> 1.$$
See \cite{alink2004diversification},  \cite{ELW09}, \cite{mainik2010optimal}, and \cite{zhu2023asymptotic} for the asymptotic superadditivity of VaR in the presence of dependence of risks.  By contrast, superadditivity of VaR in \eqref{eq:result-VaR} holds for all $p\in(0,1)$ and it is not in any asymptotic sense. 
As only minimal assumptions (i.e., monotonicity) on risk measures are imposed in our analysis,  the asymptotic superadditivity of VaR is not sufficient to obtain  the result in Proposition \ref{prop:risk-meas}, and hence it cannot be used to derive the equilibria of the risk exchange market in Section \ref{app:r1}.
\end{remark} 
%\begin{remark}[Excess-of-loss coverage on  an aggregate basis]
%Let $X,X_{1},\dots,X_{n}$ be WNAID super-Pareto losses.
%If the excess-of-loss coverage is provided on  an aggregate basis,  then stochastic dominance holds as $ X \wedge c\le_{\rm st} (\sum_{i=1}^n \theta_iX_i)\wedge c$ where $c>z_{X}$ is the threshold.
%Hence, for any weakly monotone risk measure $\rho:\mathcal X\rightarrow \R$, we have $\rho( X \wedge c )\le\rho( (\sum_{i=1}^n \theta_iX_i )\wedge c)$, and  a diversification penalty exists for $\rho$.
%Unlike the situation of model A in Proposition \ref{prop:risk-meas}, strict inequality may not hold for $\rho=\VaR_p$   because $X \wedge c$ and $(\sum_{i=1}^n \theta_iX_i)\wedge c $  have the same $p$-quantile $c$   for large $p$. Nevertheless, for the expected utility preference  
%$
%E_v,
%$
%we have 
% $$ \E[v( X \wedge c )] < \E[ v  ((\theta_1X_1+\dots+\theta_nX_n)\wedge c)],$$
% for $c>z_{X}$ and $v$ strictly increasing on $[z_{X},c]$. This is because $E_v$ is   strictly monotone in the sense that for $X\le_{\rm st}Y $ taking values in $[z_{X},c]$ and $X\not \laweq Y$, we have $E_v(X)<E_v(Y)$. 
% \end{remark}

Proposition \ref{prop:risk-meas} further leads to the following optimal decision for an agent in a market where several WNAID super-Pareto losses are present. For vectors $\mathbf x=(x_1,\dots,x_n)\in\R^n$ and $\mathbf y=(y_1,\dots,y_n)\in\R^n$, their dot product is $\mathbf x\cdot \mathbf y=\sum_{i=1}^nx_iy_i$ and %the $L^1$-norm of 
%$\mathbf x$ is
we denote by $\|\mathbf x\|=\sum_{i=1}^n|x_i|$.  Suppose that the agent needs to decide on a position $\mathbf w\in \R_+^n$ across these losses to minimize the total risk.  The agent faces a total loss 
$
\mathbf w\cdot \mathbf Y -g(\|\mathbf w\|)$ where the function $g$ represents a compensation that depends   on $\mathbf w$ through  $\|\mathbf w\|$, and 
$\mathbf Y=(Y_1,\dots,Y_n)$.
The agent's optimization problem then becomes \begin{equation}\label{eq:opt1}
\mbox{to minimize~~~} \rho\(\mathbf w\cdot \mathbf Y -g(\|\mathbf w\|) \) \mbox{~~~subject to $\mathbf w\in\mathbb R_+^n$ and $\|\mathbf w\|=w$ with given $w>0$,}
\end{equation}
or
\begin{equation}\label{eq:opt2}
\mbox{to minimize~~~} \rho\(\mathbf w\cdot \mathbf Y -g(\|\mathbf w\|) \) \mbox{~~~subject to $\mathbf w\in\mathbb R_+^n$}.
\end{equation} 
For $i\in [n]$, let $\mathbf e_{i,n} $ be the $i$th column vector of the $n\times n$ identity matrix, and 
%Let $\mathbf e_{i,n}$ be a vector on $\R^n$, such that the $i$th component of $\mathbf e_{i,n}$ is one, and all the other components are zero.  
$E_w=\{w\mathbf e_{i,n}: i \in [n]\}$ for $w\ge 0$, which represents the positions of  taking only one loss with exposure $w$.
%If the agent is equipped with a mildly monotone risk measure, Proposition \ref{prop:risk-meas} implies that she should avoid diversification. 
 \begin{proposition}\label{prop:concentrate}
  Let $\rho:\mathcal X_\rho\rightarrow \overline \R$ be  weakly monotone and $g:\R\to\R$.
 \begin{enumerate}[(i)]
 \item 
 %For model A, 
 If $\rho$ is mildly monotone, then the set of minimizers 
 of \eqref{eq:opt1} is $E_w$, and that  
 of \eqref{eq:opt2} is contained in $\bigcup_{w\in\R_+} E_w$.
% we have
% $$\argmin_{\mathbf w\in\mathbb R_+^n,~\|\mathbf w\|=w}\rho\(\mathbf w\cdot \mathbf Y -g(\|\mathbf w\|) \) =   E_w, \mbox{~~~~for $w\in \R_+$};$$ 
%  $$\argmin_{\mathbf w\in\mathbb R_+^n}\rho\(\mathbf w\cdot \mathbf Y-g(\|\mathbf w\|)\)\subseteq\bigcup_{w\in\R_+} E_w.$$
  \item  %For models A and B, %where $\p(A_1)=\dots=\p(A_n)$,  
 If \eqref{eq:opt1} has an optimizer, then it has an optimizer in $E_w$;
 if  \eqref{eq:opt2} has an optimizer, then it has an optimizer in $\bigcup_{w\in\R_+}  E_w$. 
  \end{enumerate}
 % \com{$\p(A_1)=\dots=\p(A_n)$?}
\end{proposition}

Proposition \ref{prop:concentrate} follows directly from Theorem \ref{thm:1} and Proposition \ref{prop:risk-meas}. There are almost no restrictions on $\rho$ and $g$ in Proposition \ref{prop:concentrate} other than monotonicity of $\rho$, and hence this result can be applied to many economic decision models.
This is related to the question raised in the Introduction:  By Proposition \ref{prop:concentrate}, as long as the agent's risk preference is monotone,  an agent should not diversify, under the setting of this section.
 %no-diversification is always optimal. 
 
% \begin{remark}
% Since $\ES_p$ is $\infty$ for super-Pareto losses,
% Proposition \ref{prop:concentrate}  applied to ES gives the trivial statement that every position has infinite risk. 
% %It is well-known that ES is not the right risk measure for losses without  finite mean. 
% The main context of application for Proposition \ref{prop:concentrate} should be risk measures which are finite for losses with infinite mean, such as VaR, $E_v$ with some sublinear $v$, and RVaR.
% \end{remark}
\begin{remark}\label{re:extension}
Although Propositions \ref{prop:risk-meas} and \ref{prop:concentrate} are stated for WNAID super-Pareto losses for which Theorem \ref{thm:1} can be applied, it is clear that they also hold for the following models considered in Section \ref{sec:extension}. (a) $Y_1,\dots,Y_n$ are iid and they follow a convolution of super-Pareto distributions (Proposition \ref{prop:convolution});
(b) $Y_1,\dots,Y_n$ are super-Pareto and identically distributed, and their copula is a mixture of copulas for  comonotonicity,
independence, and weak negative association with the weight on comonotonicity copula being strictly less than 1 (Proposition \ref{prop:mixture});
(c) $Y_{1},\dots,Y_{n}$ are iid random variables distributed as  $X$ beyond $x\ge z_X$ with $Y_1\ge_{\rm st} X$ where $X$ is super-Pareto, and $\rho=\VaR_p$ for $p\in[F_X(x),1)$ (Proposition \ref{prop:tail}).
%\item (Theorem \ref{prop:bounded}); 
\end{remark}

\section{Equilibrium analysis in a risk exchange economy}\label{app:r1}

 \subsection{The super-Pareto risk sharing market model}

%\subsection{The Pareto risk sharing market model}
 Suppose that there are $n\ge 2$ agents in a risk exchange market. 
 %For instance, these agents may represent insurance companies. 
  Let  $\mathbf X=(X_1,\dots,X_n)$, where $X_{1},\dots,X_{n}$ are WNAID super-Pareto random variables. The $i$th agent faces a loss $a_i X_i$, where $a_i>0$ is the initial exposure. In other words, the initial exposure vector of agent $i$ is  $\mathbf a^i=a_i \mathbf e_{i,n}$, and the corresponding loss  can be written as $\mathbf a^i \cdot\mathbf X=a_iX_i$. 
  All results in this section work for WNAID super-Pareto losses (more general than iid), but conceptually, it may be simpler (and harmless) to consider iid super-Pareto losses for an interpretation of the market.

In this risk exchange market, 
each agent decides whether and how to share the losses with the other agents. For $i\in[n]$, let $p_i\ge 0$ be the premium (or compensation) for one unit of loss $X_i$; that is, if an  agent takes $b\ge 0$ units of loss $X_i$, it receives the premium $b p_i $, which is linear in $b$. Denote by $\mathbf p=(p_1,\dots,p_n)\in \mathbb R_+^n$ the (endogenously generated) premium vector.
 Let $\mathbf w^i\in\mathbb R_+^n$ be the exposure vector of agent $i$ on $\mathbf X$ after risk sharing. The loss of agent   $i\in[n]$ after risk sharing is 
 $$L_{i}(\mathbf w^i, \mathbf p)=\mathbf w^i \cdot\mathbf X - (\mathbf w^i-\mathbf a^i)\cdot\mathbf p.$$

For each $i\in [n]$, 
assume that  agent $i$ is equipped with a risk measure $\rho_i:\X\to \R$, where $\X$ contains the convex cone generated by $\{\mathbf X\}\cup\R^n$.
Moreover,   there is a cost associated with taking a total risk position $\Vert \mathbf w^i \Vert $  different from the initial total exposure $\Vert \mathbf a^i \Vert$.
The cost is modelled by $ c_i(\Vert \mathbf w^i \Vert -\Vert \mathbf a^i \Vert) $, where $c_i$ is a non-negative convex function satisfying $c_i(0)=0$. 
Some examples of $c_i$ are $c_i(x)=0$ (no cost),  $c_i(x)=\lambda_i |x|$ (linear cost), $c_i(x)=\lambda_i x^2$ (quadratic cost), and $c_i(x)= \lambda_i x_+$ (cost only for excess risk taking), where $\lambda_i> 0$.
We denote by $c_{i-}'(x)$ and $c_{i+}'(x)$    the left and right derivatives of $c_i$ at $x\in \R$, respectively.

The above setting is called a \emph{super-Pareto risk sharing market}. 
 In this market, the goal of each agent is to choose an exposure vector so that their own risk is minimized.
 % i.e., minimizing
 % $\rho_i(L_{i}(\mathbf w^i, \mathbf p)) +  c_i(\Vert \mathbf w^i \Vert -\Vert \mathbf a^i \Vert)$ over $\mathbf w^i\in \R_+^n$, $i\in[n]$. 
 An \emph{equilibrium} of the market is 
 a tuple $\(\mathbf p^*,\mathbf w^{1*},\dots,\mathbf w^{n*}\) \in (\R_+^n)^{n+1}$ if the following  two conditions are satisfied.
\begin{enumerate}[(a)]
  \item
  Individual optimality:
 \begin{equation}\label{eq:opt-internal} \mathbf w^{i*}\in {\argmin_{\mathbf w^i\in\mathbb R_+^n}}\left \{ \rho_i\(L_{i}(\mathbf w^i, \mathbf p^*)\) + c_i(\Vert \mathbf w^i \Vert -\Vert \mathbf a^i \Vert)\right\},~~~ \mbox{for each }i\in[n].\end{equation}
  \item
 Market clearance:
  \begin{equation}\label{eq:clearance-internal}  \sum_{i=1}^{n}\mathbf w^{i*}=\sum_{i=1}^{n}\mathbf a^i.\end{equation} 
  \end{enumerate}
In this case, the vector $\mathbf p^*$ is an \emph{equilibrium price},
and $(\mathbf w^{1*},\dots,\mathbf w^{n*})$ is an \emph{equilibrium allocation}. 

Some of our results rely on a popular class of risk measures, many of which can be applied to super-Pareto losses. 
 A \emph{distortion risk measure} is defined as $\rho:\X_\rho\to \R$, via
\begin{equation}
\label{eq:distor} \rho(Y)=\int_{-\infty}^0 (h(\p(Y>x))-1)\mathrm d x+\int_0^{\infty}h(\p(Y>x))\mathrm{d}x,\end{equation}
where  $h:[0,1]\rightarrow[0,1]$, called the \emph{distortion function}, is a nondecreasing function with $h(0)=0$ and $h(1)=1$. 
The distortion risk measure $\rho$, up to sign change,  
 coincides with the \emph{dual utility} of \cite{Y87} in decision theory, and it  
  includes VaR, ES, and RVaR   as special cases (see Appendix \ref{app:A} for the definition of RVaR). Almost all  distortion risk measures  are mildly monotone (see  Proposition \ref{prop:mildly} for a precise statement). 
  We assume that $\X_\rho$ contains the convex cone generated by $\{\mathbf X\}\cup\R^n$; this always holds in case $\rho$ is VaR or RVaR, but it does not hold for $\rho$ being ES as super-Pareto losses do not have finite mean.

% Some of our results rely on a popular class of risk measures, many of which can be applied to super-Pareto losses. 
%  A \emph{distortion risk measure} is defined as $\rho:\X_\rho\to \R$, via
% \begin{equation}
% \label{eq:distor} \rho(Y)=\int_{-\infty}^0 (h(\p(Y>x))-1)\mathrm d x+\int_0^{\infty}h(\p(Y>x))\mathrm{d}x,\end{equation}
% where  $h:[0,1]\rightarrow[0,1]$, called the \emph{distortion function}, is a nondecreasing function with $h(0)=0$ and $h(1)=1$. 
% The distortion risk measure $\rho$, up to sign change,  
%  coincides with the \emph{dual utility} of \cite{Y87} in decision theory.
%  As a class of  risk measures, it  
%   includes VaR, ES, and RVaR   as special cases, and  almost all  distortion risk measures  are mildly monotone (see  Proposition \ref{prop:mildly}). 
%   We assume that $\X_\rho$ contains the convex cone generated by $\{\mathbf X\}\cup\R^n$; this always holds in case $\rho$ is VaR or RVaR, but it does not hold for $\rho$ being ES as super-Pareto losses do not have finite mean.

 \subsection{No risk exchange for super-Pareto losses}\label{sec:5}\label{sec:52}

As anticipated from Proposition \ref{prop:concentrate},   each agent's optimal strategy is to not share super-Pareto losses with others if their risk measure is mildly monotone. This observation is made rigorous in 
the following result, where we obtain a necessary condition for all possible equilibria in the market, as well as two different conditions in the case of distortion risk measures.   
%Next, we show the existence of equilibria in the market under some conditions on $\rho_1,\dots,\rho_n$.  
  As before, let $X\laweq X_1$. 
% The following proposition provides a tuple $\(\mathbf p^*,\mathbf w^{1*},\dots,\mathbf w^{n*}\)$ which forms an equilibrium for agents equipped with the same distortion risk measure.

%\iffalse
%\begin{proposition}
%In the Pareto risk sharing market, suppose that  $\alpha\in(0,1]$ and $\rho_1,\dots,\rho_n$ are mildly monotone distortion risk measures on $\X$.
% Let $p=\rho_1(X_1)=\dots=\rho_n(X_n)$, $ \mathbf p^*=(p,\dots,p)$, and $(\mathbf w^{1*},\dots,\mathbf w^{n*})$ be an $n$-permutation of $(\mathbf a^{1},\dots,\mathbf a^{n} )$.  The tuple $\(\mathbf p^*,\mathbf w^{1*},\dots,\mathbf w^{n*}\)$ is an equilibrium. 
%\end{proposition}
%
%\begin{proof}
% The clearance condition is clearly satisfied.   As distortion risk measures are translation invariant and positive homogeneous (see Section \ref{app:A} for properties of risk measures), by Proposition \ref{prop:concentrate},
%  for $i\in [n]$,
%\begin{align*}
%\min_{\mathbf w^i\in\mathbb R_+^n}\rho_i\(L_{i}(\mathbf w^i, \mathbf p^*)\)
%&=\min_{\mathbf w^i\in\mathbb R_+^n}\rho_i\(\mathbf w^i \cdot\mathbf X - (\mathbf w^i-\mathbf a^i)\cdot\mathbf p^*\)\\
%&= \min_{\|\mathbf w^i\|\in\mathbb R_+}(\rho_i\(\|\mathbf w^i\|  X_1\)-(\|\mathbf w^i\|-a_i)\rho_i(X_i))\\
%&=\rho_i(a_iX_1)=\rho_i(\mathbf a^{i}\cdot\mathbf X).
%\end{align*} 
%Note that $\rho_i\(\mathbf a^j \cdot\mathbf X - (\mathbf a^j-\mathbf a^i)\cdot\mathbf p^*\)=\rho_i(\mathbf a^{i}\cdot\mathbf X)$ for $j\in [n]$. Thus, $\mathbf a^{j}$, $j\in[n]$, minimizes the risk of the $i$th agent for $i\in[n]$. 
%\end{proof}
%\fi
 
%\begin{proposition}\label{prop:equil}
% 
%\end{proposition}
%
%

\begin{theorem}\label{th:equil}
In a super-Pareto risk sharing market, suppose that  $\rho_1,\dots,\rho_n$ are mildly monotone.
\begin{enumerate}[(i)]
\item  All equilibria $\(\mathbf p^*,\mathbf w^{1*},\dots,\mathbf w^{n*}\)$ (if they exist) satisfy $\mathbf p^*=\( p,\dots,p\)$ for some $p\in\R_+$ and $\(\mathbf w^{1*},\dots,\mathbf w^{n*}\)$ is an $n$-permutation of $(\mathbf a^1,\dots,\mathbf a^n)$.
\item Suppose that   $\rho_1,\dots,\rho_n$ are  distortion risk measures on $\X$.
The tuple $\((p,\dots,p) , \mathbf a^{1},\dots,\mathbf a^{n} \)$ is an equilibrium if   $ p$ satisfies 
\begin{equation}
\label{eq:convex}
c_{i+}'(0) \ge  p-\rho_i(X)  \ge  c_{i-}'(0)~~~~\mbox{for $i\in [n]$}.
\end{equation}
\item Suppose that   $\rho_1,\dots,\rho_n$ are  distortion risk measures on $\X$. If $(p,\dots,p)$  is an equilibrium price, then 
\begin{equation}
\label{eq:convex2}
\max_{j\in [n]}c_{i+}'(a_j-a_i) \ge  p-\rho_i(X)  \ge \min_{j\in [n]}  c_{i-}'(a_j-a_i)~~~~\mbox{for $i\in [n]$}.
\end{equation}
\end{enumerate}

\end{theorem}

 Theorem \ref{th:equil} (i) states that, even if there is some risk exchange in an equilibrium, the agents merely exchange positions entirely instead of sharing a pool.
This observation is consistent with Theorem \ref{thm:1}, which implies that diversification among multiple super-Pareto losses increases risk in a uniform sense.
As there is no diversification in the optimal allocation for each agent,   taking any of these WNAID losses is equivalent for the agent, and the equilibrium price should be identical across losses.
    Part (ii) suggests that 
  if $c_i$ has a kink at $0$, i.e., $c_i'(0+)>0>c_i'(0-)$, then $p$ can be an equilibrium price if it is very close to $\rho_i(X)$ in the sense of \eqref{eq:convex}. 
Conversely, in part (iii),  if $p$ is an equilibrium price, then it needs to be close to $\rho_i(X)$ for $i\in [n]$ in the sense of \eqref{eq:convex2}. 
 This observation is quite intuitive because  by (i), the agents will not share losses but rather keep one of them in an equilibrium. 
  If the price of taking one unit of the loss is too far away from an agent's assessment of the loss, it may have an incentive to move away, and the equilibrium is broken.

% Further results on the super-Pareto risk sharing market \textcolor{blue}{are in Section} \ref{app:r1}. 
%  Section \ref{app:discuss-1} contains some further discussions on Theorem \ref{th:equil}.

% \subsection{Discussions on Theorem \ref{th:equil}}
% \label{app:discuss-1}
%    We discuss the results in Theorem \ref{th:equil} in this section.
   
The equilibrium price $p$ should be very close to the individual risk assessments, and hence the risk sharing mechanism does not benefit the agents. Indeed,  in (ii), the equilibrium allocation is equal to the original exposure, and there is no welfare gain.  
This is drastically different from a market considered in Section \ref{sec:53} below; all agents will benefit from transferring some losses to an external market (see Theorem \ref{th:external}).

In general, \eqref{eq:convex} and \eqref{eq:convex2} are not equivalent, but in the two cases below, they are: (a)  $a_1=\dots=a_n$;
  (b) $c_1=\dots=c_n=0$. 
 In either case, both \eqref{eq:convex} and \eqref{eq:convex2} are    a necessary and sufficient  condition for $(p,\dots,p)$ to be an equilibrium price. 
Hence, the  tuple $\(\mathbf p^*,\mathbf w^{1*},\dots,\mathbf w^{n*}\)$ is an equilibrium if and only if \eqref{eq:convex} holds and $(\mathbf w^{1*},\dots,\mathbf w^{n*})$ is an $n$-permutation of $(\mathbf a^{1},\dots,\mathbf a^{n} )$, which can be checked by Theorem \ref{th:equil} (i).
In case (a), $p$ cannot be too far away from $\rho_i(X)$ for each $i\in [n]$. In case (b),   $p=\rho_1(X)=\dots=\rho_n(X)$, and
an equilibrium can only be achieved when all agents agree on the risk of one unit of the loss and use this assessment for pricing.  

\begin{example}[Equilibrium for Pareto losses and VaR agents with no costs]
Suppose that $c_i=0$ for $i\in [n]$ and $X_1,\dots,X_n\sim \mathrm{Pareto}(\alpha)$, $\alpha\in(0,1]$. 
Let $\rho_i=\VaR_{q}$, $q\in(0,1)$, $i\in[n]$. The tuple $\(\mathbf p^*,\mathbf w^{1*},\dots,\mathbf w^{n*}\)$ is an equilibrium where  
$ \mathbf p^*=((1-q)^{-1/\alpha},\dots,(1-q)^{-1/\alpha})$,
and $(\mathbf w^{1*},\dots,\mathbf w^{n*})$ is an $n$-permutation of $(\mathbf a^{1},\dots,\mathbf a^{n} )$. For $i \in [n]$, $\rho_i\(L_{i}(\mathbf w^{i*}, \mathbf p^*)\)=\VaR_q(a_iX)=a_i (1-q)^{-1/\alpha}$.
\end{example}

 % We present in Section \ref{sec:53} a market with external agents.
 %  It turns out that in the presence of the external market, the picture is drastically different from the one in this section: All agents will benefit from transferring some losses to an external market; see Theorem \ref{th:external}. 
  
% Although the agents will not benefit from  losses with infinite mean like super-Pareto losses,  the situation changes if 
% the losses have finite mean, which will be discussed in Section \ref{sec:54}. 

We  offer a few  further technical remarks on Theorem \ref{th:equil}. First, 
Theorem \ref{th:equil} (ii)  and (iii) remain valid for all mildly monotone, translation invariant, and positively homogeneous risk measures. 
Second, if the range of $\mathbf w^i =(w_1^i,\dots,w_n^i)$ in \eqref{eq:opt-internal} 
 is constrained to $0\le w^i_j \le a_j $ for $j\in [n]$, then   $\((p,\dots,p) , \mathbf a^{1},\dots,\mathbf a^{n} \)$  in Theorem \ref{th:equil} (ii) is still an equilibrium under the condition \eqref{eq:convex}. However, the characterization statement in (i)  is no longer guaranteed, which can be seen from the proof of Theorem \ref{th:equil} in Section \ref{app:proof}. As a result, (iii) cannot be obtained either. 
 Third, the super-Pareto risk sharing market is closely related to some models considered in Section \ref{sec:extension} (see Remark \ref{re:extension}). Since these models have similar results to Theorem \ref{thm:1}, which is used to establish Theorem \ref{th:equil}, we can check that  the equilibrium in Theorem \ref{th:equil} (ii) still holds under these models.
% \item 
% The super-Pareto risk sharing market is closely related to model A in Section \ref{sec:4}.
% Since model B has similar properties to model A in Proposition \ref{prop:concentrate}, we can check that  the equilibrium in Theorem \ref{th:equil} (ii) still holds if we replace model A by model B, where the triggering events have the same probability of occurrence (i.e., $\p(A_1)=\dots=\p(A_n)$). However, we cannot guarantee that all equilibria for model B have the form in (i) since holding one of the super-Pareto risks may not be the only optimal strategy for agents in model B; see Proposition \ref{prop:concentrate}. 

\subsection{A market with external risk transfer}\label{sec:53}

In Section \ref{sec:52}, we have considered risk exchange among agents with initial super-Pareto losses.
 Next, we consider an extended market with external agents to which risk can be transferred with compensation from the internal agents. %The $n$ agents in the previous setting will be called internal agents. 

By Theorem \ref{th:equil}, agents cannot reduce their risks by sharing super-Pareto losses within the group. As such, they may seek to transfer their risks to external parties. In this context, the internal agents are risk bearers, and the external agents are  institutional investors without initial position of  super-Pareto losses.
% and the external agents are insurance companies or institutional investors (e.g., investors in CAT bonds). 
% For example, catastrophe bond (CAT bond) is one of the most common alternative risk transfer mechanisms to raise money to cover catastrophic losses of insurers. 
% In this context, the external agents are investors in CAT bonds. 

%The raised money is invested to generate proceeds. If a triggering event happens (e.g., the loss is above some threshold), a portion or all of the fund will be paid to the insurers.
%If the bond is not triggered, funds will be fully returned to the investors with a positive yield. For simplicity, we assume that CAT bonds are issued by insurers, and each CAT bond can only transfer the losses of the issuer.

 Consider   a super-Pareto risk sharing market with $n$ internal agents and   $m\ge 1$ external agents
 equipped with %
 the same  
 risk measure $\rho_{E}:\mathcal X\rightarrow \R$.
  Let $\mathbf u^j\in\mathbb R_+^n$ be the exposure vector of external agent $j\in[m]$ after sharing the risks of the internal agents. For external agent $j$, the
 loss for taking position $\mathbf u^j$ is 
  $$L_{E}(\mathbf u^j, \mathbf p)=\mathbf u^j \cdot\mathbf X -  \mathbf u^j \cdot\mathbf p,$$ where $\mathbf p=(p_1,\dots,p_n)$ is the premium vector.  
Like the internal agents, the goal of the external agents is to minimize their risk plus cost. That is, for $j\in [m]$, external agent $j$ minimizes $%r_{E,j}\(\mathbf u^j\)=
\rho_{E}\(L_{E}(\mathbf u^j, \mathbf p) \)+c_{E}( \|\mathbf u^j\| )$, 
where   $c_{E}$ is a non-negative cost function satisfying $c_{E}(0)=0$.

For tractability, 
we will also make  some simplifying assumptions on the internal agents. 
We assume that the internal agents  have the same  risk measure $\rho_I$ and the same cost function $c_I$. 
%Denote by $\rho_I$, $\rho_E$, $c_I$ and $c_E$ the risk measures and cost functions of the internal (labelled by $I$) and external (labelled by $E$) agents, respectively.  That is,  $\rho_I=\rho_1=\dots=\rho_n$, $\rho_E=\rho_{E,1}=\dots=\rho_{E,m}$, $c_I=c_1=\dots=c_n$ and $c_E=c_{E,1}=\dots=c_{E,m}$.
Assume that $c_I$ and $c_E$ are strictly convex and continuously differentiable except at $0$,
and   $\rho_I$ and $\rho_{E}$ are mildly monotone distortion risk measures defined on $\X$.
In addition, all internal agents have the same amount  $a>0$ of initial loss exposures, i.e., $a=a_1=\dots=a_n$.
Finally,  we consider the situation where the number of external agents is larger than the number of internal agents  by assuming that $m=kn$, where $k$ is a positive integer, possibly large. 

An \emph{equilibrium} of this market  is 
 a tuple $(\mathbf p^*, \mathbf w^{1*},\dots,\mathbf w^{n*}, \mathbf u^{1*},\dots,\mathbf u^{m*} ) \in (\R_+^n)^{n+m+1}$ if the following  two conditions are satisfied.
\begin{enumerate}[(a)]
  \item
  Individual optimality:
  \begin{align}\label{eq:opt-external1}
&  \mathbf w^{i*}\in {\argmin_{\mathbf w^i\in\mathbb R_+^n}}\left\{ \rho_I\(L_{i}(\mathbf w^i, \mathbf p^*)\)+c_{I}( \|\mathbf w^i\| -\|
  \mathbf a^i\|)\right\},~~~ \mbox{for each }i\in[n];
\\
 \label{eq:opt-external2}
&  \mathbf u^{j*}\in {\argmin_{\mathbf u^j\in\mathbb R_+^n}}\left\{  \rho_{E}\(L_{E}(\mathbf u^j, \mathbf p^*)    \)+ c_{E}( \|\mathbf u^j\| )\right\},~~~ \mbox{for each }j\in[m].
  \end{align}

  \item
 Market clearance:
 \begin{equation}\label{eq:clearance-external}
  \sum_{i=1}^{n}\mathbf w^{i*}+\sum_{j=1}^{m}\mathbf u^{j*}=\sum_{i=1}^{n}\mathbf a^i.
  \end{equation}
  \end{enumerate} 
  The vector $\mathbf p^*$ is an \emph{equilibrium price},
and $(\mathbf w^{1*},\dots,\mathbf w^{n*})$ and $(\mathbf u^{1*},\dots,\mathbf u^{m*})$ are \emph{equilibrium allocations} for the internal and external agents, respectively.
  Before identifying  the equilibria  in this market, 
we first make some simple observations.
Let  $$L_{E}(b)=c_{E}'(b)  + \rho_E(X) \mbox{~~~and~~~}L_{I}(b)=c_{I}'(b)  + \rho_I(X),~~~~~~~b\in \R.$$
We will write $L_{I}^-(0)=c_{I-}'(0)  + \rho_I(X)$
and   $L_{I}^+(0)=c_{I+}'(0)  + \rho_I(X)$
to emphasize that the left and right derivative of $c_I$ may not coincide at $0$; this is particularly relevant in~Theorem \ref{th:equil} (ii). On the other hand, $L_E(0)$ only has one relevant version since the allowed position is non-negative. Note that both $L_E$ and $L_I$ are  continuous except at $0$ and strictly increasing.

If an external agent takes only one source  of loss (intuitively optimal from Proposition \ref{prop:concentrate}) among $X_1,\dots,X_n$ (we use the generic variable $X$ for this loss), then $L_E(b)$ is the marginal cost 
of further increasing their position at $b X$. As a compensation, this agent will also receive $p$. Therefore, the external agent has incentives to participate in the risk sharing market if
 $p > L_{E}(0)$. If  $p \le L_{E}(0)$, due to the strict convexity of $c_E$, this agent will not take any risks.
 On the other hand, if $p \ge  L^-_{I}(0)$, which means that it is expensive to transfer the loss externally, then the internal agent has no incentive to transfer.
For a small risk exchange to benefit both parties, 
  we need $L_{E }(0)< p< L_{I}^-(0)$.
  This implies, in particular,  $$\rho_E(X)\le L_{E }(0) < p< L_{I}^-(0) \le  \rho_I(X),$$ which means that the risk is more acceptable to the 
  external agents than to the internal agents, and the price is somewhere  between the two risk assessments.
  The above intuition is helpful to understand the conditions in the following theorem. Denote by $\mathbf 0_n=(0,\dots,0)\in\R^n$.
%  This is consistent 

 %We consider the situation $L_{E+}\le L_{I-}$ and $L_{E+}>L_{I-}$ separately. 

\begin{theorem}\label{th:external}
Consider the super-Pareto risk sharing market of $n$ internal  and $m=kn$ external agents. 
Let $\mathcal E=(\mathbf p , \mathbf w^{1*},\dots,\mathbf w^{n*}, \mathbf u^{1*},\dots,\mathbf u^{m*} ) $.
\begin{enumerate}[(i)]
%\item If 
%  $ 
%L_E(a/k)  < L_I(-a),  
%  $ 
%  then there is no equilibrium.
\item 
Suppose that  
  $ 
L_E(a/k)  < L_I(-a).  $ The tuple $\mathcal E$ is an equilibrium if and only if   
$\mathbf p=(p,\dots,p)$,  $p=L_E(a/k)$,
$(\mathbf u^{1*},\dots,\mathbf u^{m*})$ is a permutation of $u^*( \mathbf e_{\lceil 1/k\rceil ,n},\dots, \mathbf e_{\lceil m/k\rceil ,n})$, $u^*=a/k$,    
and $(\mathbf w^{1*},\dots,\mathbf w^{n*})=(\mathbf 0_n,\dots,\mathbf 0_n)$.  
%  \item  Suppose that $L_E(a/k)  \ge  L_I(-a)$ and   $L_{E}(0) < L_{I}^-(0)$. Let  $u^*$ be the unique solution to 
%  \begin{equation}\label{eq:external-price}
% L_E(u) = L_I(-ku), ~~~~~u\in (0,a/k].
%  \end{equation}
% The tuple  $\mathcal E$  is an equilibrium if and only if    
%$\mathbf p=(p,\dots,p)$,    $p= L_E(u^*) $, 
%$(\mathbf u^{1*},\dots,\mathbf u^{m*})$ is a permutation of $ u^* ( \mathbf e_{\lceil 1/k\rceil ,n},\dots, \mathbf e_{\lceil m/k\rceil ,n})$,    
%and $(\mathbf w^{1*},\dots,\mathbf w^{n*})$ is a permutation of $ (a-ku^*)( \mathbf e_{1,n},\dots, \mathbf e_{n,n})$.
\item  Suppose that $L_E(a/k)  \ge  L_I(-a)$ and   $L_{E}(0) < L_{I}^-(0)$. Let  $u^*$ be the unique solution to 
  \begin{equation}\label{eq:external-price}
 L_E(u) = L_I(-ku), ~~~~~u\in (0,a/k].
  \end{equation}
 The tuple  $\mathcal E$  is an equilibrium if and only if    
$\mathbf p=(p,\dots,p)$,    $p= L_E(u^*) $, $(\mathbf u^{1*},\dots,\mathbf u^{m*})= u^* ( \mathbf e_{k_1 ,n},\dots, \mathbf e_{k_m ,n})$, and $(\mathbf w^{1*},\dots,\mathbf w^{n*})= (a-ku^*)( \mathbf e_{\ell_1,n},\dots, \mathbf e_{\ell_n,n})$, where $k_1,\dots,k_m\in[n]$ and $\ell_1,\dots,\ell_n\in[n]$ such that
$u^*\sum_{j=1}^m\id_{\{k_j=s\}}+(a-ku^*)\sum_{i=1}^n\id_{\{\ell_i=s\}}=a$ for each $s\in[n]$.

Moreover, if $u^*< a/(2k)$, then the tuple  $\mathcal E$  is an equilibrium if and only if    
$\mathbf p=(p,\dots,p)$,    $p= L_E(u^*) $, 
$(\mathbf u^{1*},\dots,\mathbf u^{m*})$ is a permutation of $ u^* ( \mathbf e_{\lceil 1/k\rceil ,n},\dots, \mathbf e_{\lceil m/k\rceil ,n})$,    
and $(\mathbf w^{1*},\dots,\mathbf w^{n*})$ is a permutation of $ (a-ku^*)( \mathbf e_{1,n},\dots, \mathbf e_{n,n})$.
\item Suppose that   $L_{E}(0)\ge L_{I}^-(0)$.
 The tuple $\mathcal E$ is an equilibrium if and only if   
$\mathbf p=(p,\dots,p)$, 
$p\in [L_{I}^-(0),  L_{E}(0)\wedge L_{I}^+(0)]$, 
$(\mathbf u^{1*},\dots,\mathbf u^{m*})=(\mathbf 0_n,\dots,\mathbf 0_n)$,    
and $(\mathbf w^{1*},\dots,\mathbf w^{n*})$ is a permutation of $ a( \mathbf e_{1,n},\dots, \mathbf e_{n,n})$.  

\end{enumerate}

\end{theorem}

% To interpret Theorem \ref{th:external} (i), note that $L_E(a/k)  < L_I(-a)$ implies  $L_E(u)  < L_I(w-a)$ for all $u\in[0,a/k]$ and  $w\in[0,a]$. It means that if the price of transferring a unit of risk is in $[L_E(a/k),L_I(-a)]$, the optimal position for each internal agent will be 0, and the external agents will have the incentives to increase their exposures from 0 to more than $a/k$. In this case, the individual optimality conditions \eqref{eq:opt-external1} and \eqref{eq:opt-external2}  and the clearance condition \eqref{eq:clearance-external} cannot be satisfied at the same time. Therefore, there is no equilibrium. 

 Compared with Theorem \ref{th:equil}, where  no benefits exist from risk sharing among the internal agents, Theorem \ref{th:external} (ii) implies that  in the presence of external agents, every party in the market may get better from risk sharing. More specifically, if $L_{E}(0) < L_{I}^-(0)$, (i.e., the marginal cost of increasing an external agent's position from 0 is smaller than the marginal benefit of decreasing an internal agent's position from $a$), there exists an equilibrium price $p\in[L_{E}(0), L_{I}^-(0)]$ such that all parties in the market can improve their objectives. 
 % The condition $L_{E}(0) < L_{I}^-(0)$ is crucial to such a win-win situation, as a price less than $L_{I}^-(0)$ will motivate the internal agents to transfer risk, and a price greater than $L_{E}(0)$ will motivate the external agents  to receive risks.
 Moreover, if $u^*< a/2k$, i.e, the optimal position of each external agent is very small compared with the total position of each loss in the market, the loss $X_i$ for each $i\in[n]$, has to be shared by one internal agent and $k$ external agents to achieve an equilibrium.
  Theorem \ref{th:external} (i) shows that if $L_E(a/k)  < L_I(-a)$, all the losses will be transferred to the external agents.
  Theorem \ref{th:external} (iii) shows that if $L_{E}(0) \ge L_{I}^-(0)$, no agent will share risks.

We make further observations on Theorem \ref{th:external} (ii). From \eqref{eq:external-price}, it is straightforward to see that if $k$ gets larger,
%(more external agents are in the market), 
the equilibrium price $p$ gets smaller. Intuitively, if more external agents are willing to take risks, they have to compromise on the received compensation to get the amount of risks they want. The lower price further motivates the internal agents to transfer more risks to the external agents. Indeed, by \eqref{eq:external-price}, $ku^*$ gets larger as $k$ increases. On the other hand, $u^*$ gets smaller as $k$ increases. In the equilibrium model, each external agent will take less risk if more external agents are in the market. These observations can be seen more clearly in the example below.

%The equilibrium in Theorem \ref{th:external} is clearly invariant under permutations since all internal agents have the same objectives and so do all external agents. 
%Indeed, 
%any equilibrium in this market
%has the form in Theorem \ref{th:external}  up to a permutation. 

\begin{example}[Quadratic cost]
Suppose that the conditions in Theorem \ref{th:external} (ii) are satisfied (this implies $\rho_E(X)<\rho_I(X)$ in particular), $c_{I}(x)=\lambda_I  x^2$, and $c_{E}(x)=\lambda_E  x^2$, $x\in\R$, where $ \lambda_I,\lambda_E > 0$. We can compute the equilibrium price
%$u^*(p^*)=\(p^*-\rho_{E}(X)\)/2\lambda_E$ and $w^*(p^*)=\(p^*-\rho_{I}(X)\)/2\lambda_I+a$, where
$$p=\frac{k\lambda_I}{k\lambda_I+\lambda_E}\rho_E(X)+\frac{\lambda_E}{k\lambda_I+\lambda_E}\rho_I(X).$$
Therefore, the equilibrium price is a weighted average of $\rho_E(X)$ and $\rho_I(X)$, where the weights depend on $k$, $\lambda_I$, and $\lambda_E$. 
We also have the equlibrium allocations $\mathbf u^*=(u,\dots,u)$ and $\mathbf w^*=(w,\dots,w)$ where
$$u=\frac{\rho_I(X)-\rho_E(X)}{2(k\lambda_I+\lambda_E)}~~~~\mbox{and}~~~~w=\frac{k(\rho_E(X)-\rho_I(X))}{2(k\lambda_I+\lambda_E)}+a.$$
% It is clear that $p$ moves in the opposite direction of $k$. Moreover, if more external agents are in the market, each external agent will take fewer losses, while each internal agent will transfer more losses to the external agents.
It is clear that the above observations on Theorem \ref{th:external} (ii) hold in this example. Moreover,
if  $\lambda_I$ increases, the internal agents will be less motivated to transfer their losses. To compensate for the increased penalty, the price paid by the internal agents will decrease so that they are still willing to share risks to some extent. The interpretation is similar if $\lambda_E$ changes. Although the increase of different penalties ($\lambda_E$ or $\lambda_I$) have different impacts on the price, the increase of either $\lambda_E$ or $\lambda_I$ leads to less incentives for the internal and external agents to participate in the risk sharing market.  
\end{example}

\subsection{Risk exchange for  losses with finite mean: A contrast}
\label{sec:54}
In contrast to the settings in Sections \ref{sec:52}  and \ref{sec:53}, 
we study   losses with finite mean below, for the purpose of providing a constrast. Consider a market which is the same as the super-Pareto risk sharing market except that the losses are iid with finite mean. This market  is called a \emph{risk sharing market with finite mean}.
The following proposition shows that agents prefer to share finite-mean losses among themselves if they are equipped with ES. 
\begin{proposition}\label{prop:equil-ES}
In a  risk sharing market with finite mean, suppose that  $\rho_1=\dots=\rho_n=\ES_q$ for some $q\in (0,1)$.  Let 
$$\mathbf w^{i*}=\frac{a_i}{\sum_{j=1}^n a_j }\sum_{j=1}^{n}\mathbf a^j \mbox{ for $i\in[n]$~~~ and ~~~}
\mathbf p^*=\(\E\[X_1|A\],\dots,\E\[X_n| A\]\),$$ where $A=\{\sum_{i=1}^n a_iX_i\ge \VaR_q\(\sum_{i=1}^n a_iX_i\)\}$.
Then the tuple $\(\mathbf p^*,\mathbf w^{1*},\dots,\mathbf w^{n*}\)$ is an equilibrium.
%If $c_1,\dots,c_n$ are positive  except at $0$, then the equilibrium is unique.
\end{proposition}

A sharp contrast is visible between the equilibrium in Theorem \ref{th:equil} and that in Proposition \ref{prop:equil-ES}. For WNAID super-Pareto losses, which do not have finite mean, the equilibrium price is the same across individual losses, and  agents do not share losses at all.
For iid finite-mean losses and ES agents, each individual loss has a different equilibrium price, and  agents share all losses proportionally. 

We choose the risk measure ES here because it leads to an explicit expression of the equilibrium. Although ES is not finite for super-Pareto losses (thus, it does not fit Theorem \ref{th:equil}),
it can be approximated arbitrarily closely by RVaR (e.g., \cite{ELW18}) which fits the condition of Theorem \ref{th:equil}. By this approximation, the observation that  agents prefer diversification in Proposition \ref{prop:equil-ES} may hold if ES is replaced by RVaR, although we do not have an explicit result. Below, we provide an example of two agents with normal random variables as their risks.
%The economic interpretations remain valid if we use RVaR in Proposition \ref{prop:equil-ES}, but we do not have an explicit form.
 \begin{example}
In a risk sharing market with finite mean, suppose that there are two agents with $X_1,X_2\sim \mathrm{N}(0,1)$ being independent and $a_1=a_2=1$. Let $\rho_1=\rho_2=\RVaR_{p,q}$ where $0\leq p<q<1$. For $X\sim \mathrm{N}(\mu,\sigma^2)$, by using results for ES in Example 2.14 of \cite{MFE15}, we have
 $$\RVaR_{p,q}(X)=\mu+\sigma C_{p,q},
 \mbox{~where~} C_{p,q}=\frac{\phi(\Phi^{-1}(p))-\phi(\Phi^{-1}(q))}{q-p}.$$
 Let $\mathbf w^i=(w^i_1,w^i_2)$ for $i=1,2$ and $\mathbf p^*=(p^*,p^*)=(C_{p,q}/\sqrt{2},C_{p,q}/\sqrt{2})$. Agent $i\in \{1,2\}$  aims to minimize
 \begin{align*}
\rho_i\(L_{i}(\mathbf w^i, \mathbf p^*)\) + c_i(\Vert \mathbf w^i \Vert -\Vert \mathbf a^i \Vert)&=\RVaR_{p,q}(\mathbf w^i \cdot\mathbf X)- (\mathbf w^i-\mathbf a^i)\cdot\mathbf p^*+ c_i(\Vert \mathbf w^i \Vert -\Vert \mathbf a^i \Vert)\\
&=C_{p,q}\sqrt{(w_1^i)^2+(w_2^i)^2}- w_1^ip^*-w_2^ip^*+p^*+ c_i(w_1^i+w_2^i-1).
 \end{align*}
  Let 
  $r(x,y)=C_{p,q}\sqrt{x^2+y^2}- xp^*-yp^*=  p^* \sqrt{2x^2+2y^2} - p^*(x+y)\ge 0$ for  $(x,y)\in\R_+^2.$ It is easy to verify that $r$ is minimized when $x=y$, with $r(x,x)=0$. Moreover, $c_i(w_1^i+w_2^i-1)$ is minimized when $w_1^i+w_2^i=1$.  Therefore, $(\mathbf p^*, (0.5,0.5), (0.5,0.5))$ is an equilibrium of this market.
\end{example}

% \begin{remark}
% Proposition \ref{prop:equil-ES}, assuming iid losses with finite mean,  works for all convex risk measures. The intuition is that the value of convex risk measures can be reduced by diversification, i.e., $\rho(\lambda X+(1-\lambda)Y)\le\lambda\rho(X)+(1-\lambda)\rho(Y)$ where $\rho$ is a convex risk measure, $X$ and $Y$ are two random variables with finite mean, and $\lambda\in(0,1)$.  
% % Convex risk measures are not suitable for the case of super-Pareto risks as they will always be infinite for risks without finite mean (see e.g., \cite{filipovic2012canonical}).
% \end{remark}

 %, but a general conclusion  is not available. 
 \section{Some simple  examples based on real data}\label{sec:6}

 \subsection{Extremely heavy-tailed Pareto losses}\label{sec:62}

 In addition to the examples mentioned in the Introduction,  
we provide two further data examples: the first one on marine losses, and the second one on suppression costs of wildfires. The marine losses dataset, from the insurance data repository \texttt{CASdatasets},\footnote{Available at http://cas.uqam.ca/.} was originally collected by a French private insurer and comprises 1,274 marine losses (paid)
between January 2003 and June 2006. The wildfire dataset\footnote{See https://wildfire.alberta.ca/resources/historical-data/historical-wildfire-database.aspx.} contains 10,915 suppression costs in Alberta, Canada from 1983 to 1995. For the purpose of this section, we only provide the Hill estimates of these two datasets although a more detailed EVT analysis is available (see \cite{MFE15}). The Hill estimates of the tail indices $\alpha$ are presented in Figure \ref{fig:real-data-2}, where the black curves represent the point estimates and the red curves represent the $95\%$ confidence intervals with varying thresholds; see \cite{MFE15} for more details on the Hill estimator. As suggested by \cite{MFE15}, one may roughly chose a threshold around  the top $5\%$ order statistics of the data. Following this suggestion, the tail indices $\alpha$ for the marine losses and wildfire suppression costs are estimated as $0.916$ and $0.847$ with $95\%$ confidence intervals being $(0.674, 1.158)$ and $(0.776, 0.918)$, respectively;  thus, these losses/costs have infinite mean if they follow Pareto distributions in their tails regions.
\begin{figure}[htb]
    \begin{minipage}[t]{.49\textwidth}
        \centering
        \includegraphics[width=\textwidth]{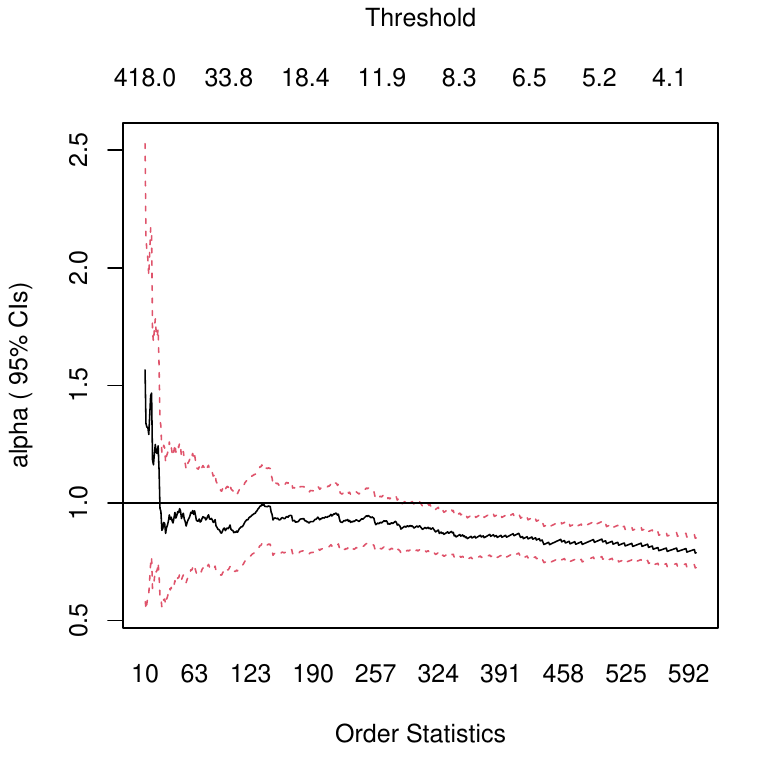}
        \subcaption{Marine losses}\label{fig:1}
    \end{minipage}
    \hfill
    \begin{minipage}[t]{.49\textwidth}
        \centering
        \includegraphics[width=\textwidth]{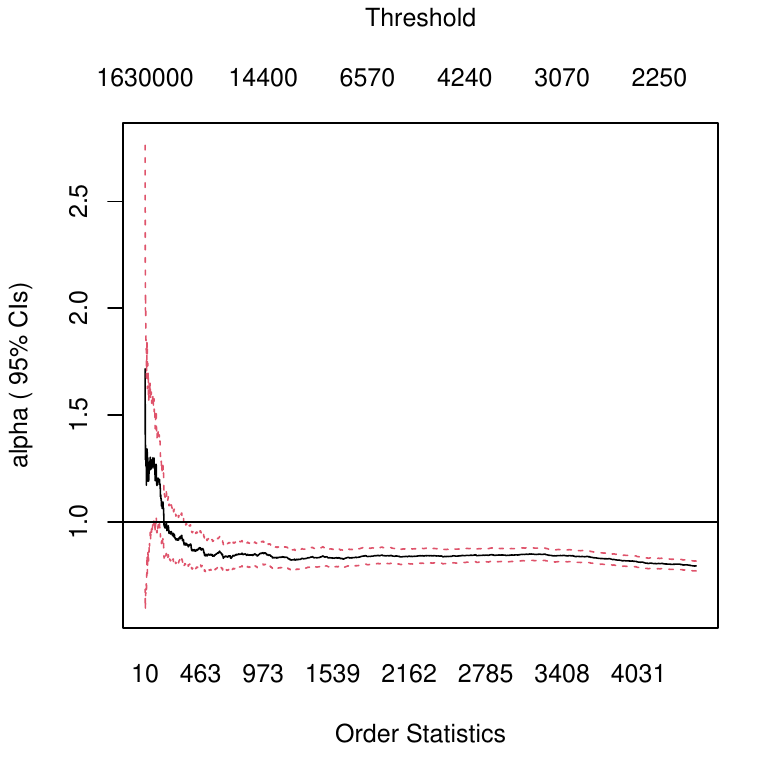}
        \subcaption{Wildfire suppression costs}\label{fig:2}
    \end{minipage}
    \caption{Hill plots  for the marine losses and wildfire suppression costs: For each risk, the Hill estimates are plotted as black curve with the $95\%$ confidence intervals being red curves.}\label{fig:real-data-2}
\end{figure}

These observations suggest that the two loss datasets may have similar tail parameters. 
%As discussed in Remark \ref{remark:GPD},  Lemma \ref{thm:1} can be applied to generalized Pareto distributions.
As one example of super-Pareto distributions,  the generalized Pareto distribution when $\xi\ge1$, is given  by 
  \begin{equation*}%\label{eq:GPD}
  G_{\xi,\beta}(x)=1-\(1+\xi\frac{x}{\beta}\)^{-1/\xi},~~~ x\ge 0,
  \end{equation*}
where   $\beta>0$. 
By Theorem \ref{thm:1}, if two loss random variables $X_1$ and $X_2$ are independent and follow generalized Pareto distributions with the same tail parameter $\alpha=1/\xi<1$, 
then, for all $p\in (0,1)$,
 \begin{align}
 \VaR_p(X_1+X_2)>\VaR_p(X_1)+\VaR_p(X_2).\label{eq:real-data}
 \end{align}  
 Even if $X_1$ and $X_2$ are not Pareto distributed, as long as their tails are Pareto, \eqref{eq:real-data} may hold for $p$ relatively large, as suggested by Proposition \ref{prop:tail}.

  We will verify \eqref{eq:real-data} on our datasets to show how the implication of Theorem \ref{thm:1} holds for real data.  
  Since the marine losses data were scaled to mask the actual losses,
 we renormalize it by multiplying the data by 500 to make it roughly on the same scale as that of the wildfire suppression costs;\footnote{The average marine losses (renormalized) and the average wildfire suppression costs are $12400$ and $12899$.} this normalization is made only for better visualization. Let  $\widehat F_1$ be
the empirical distribution  of the marine losses (renormalized) and $\widehat F_2$ be the empirical distribution  of the wildfire suppression costs. Take independent random variables $\widehat Y_1\sim \widehat F_1$ and $\widehat Y_2\sim \widehat F_2$. Let $\widehat F_1 \oplus \widehat F_2$ be the distribution with quantile function $p\mapsto \VaR_p(\widehat Y_1)+\VaR_p(\widehat Y_2)$, i.e., the comonotonic sum, and $\widehat F_1 * \widehat F_2$ be the distribution of $\widehat Y_1+\widehat Y_2$, i.e., the independent sum.

 The differences between the  distributions  $\widehat F_1 \oplus \widehat F_2$ and $\widehat F_1 * \widehat F_2$ can be seen in Figure \ref{fig:ecdfdiff}. We observe that $\widehat F_1 * \widehat F_2$ is less than $\widehat F_1 \oplus \widehat F_2$ over a wide range of loss values. In particular, the relation holds for all losses less than 267,659.5 (marked by the vertical line in Figure \ref{fig:ecdfdiff}). Equivalently, we can see from Figure \ref{fig:qsamples} that 
 \begin{align} \VaR_p(\widehat Y_1+\widehat Y_2)> \VaR_p(\widehat Y_1)+\VaR_p(\widehat Y_2)
\label{eq:real-data-2}
 \end{align}   holds  
unless $p$ is greater than $0.9847$ (marked by the vertical line in Figure \ref{fig:qsamples}). 
Recall that $\widehat F_1 * \widehat F_2\le \widehat F_1 \oplus \widehat F_2 $ is equivalent to  \eqref{eq:real-data-2} holding for all $p\in(0,1)$. Since the quantiles are directly computed from   data, thus from  distributions with bounded supports, for $p$ close enough to $1$ it must hold  that 
$ \VaR_p(\widehat Y_1+\widehat Y_2)\le \VaR_p(\widehat Y_1)+ \VaR_p(\widehat Y_2)$. 
Nevertheless,  we observe  \eqref{eq:real-data-2} for most values of $p\in(0,1)$. 
Note that the observation of \eqref{eq:real-data-2} is entirely empirical and it does not use any fitted models. 

Let $F_1$ and $F_2$ be the true distributions (unknown) of the marine losses (renormalized) and wildfire suppression costs, respectively. We are interested in whether the first-order stochastic dominance relation $F_1 *  F_2\le F_1 \oplus  F_2 $ holds. Since we do not have access to the true distributions, we generate two independent random samples of size $10^4$ (roughly equal to the sum of the sizes of the datasets, thus with a similar magnitude of randomness) from the distributions $\widehat F_1 \oplus \widehat F_2$ and $\widehat F_1 * \widehat F_2$. We treat these samples as independent random samples from $ F_1 \oplus  F_2$ and $ F_1 *  F_2$ and test the hypothesis using Proposition 1 of \cite{barrett2003consistent}. The p-value of the test is greater than $0.5$ and we are not able to reject the hypothesis $F_1 *  F_2\le F_1 \oplus  F_2 $.

\begin{figure}[t]
    \begin{minipage}[t]{.49\textwidth}
        \centering
        \includegraphics[width=\textwidth]{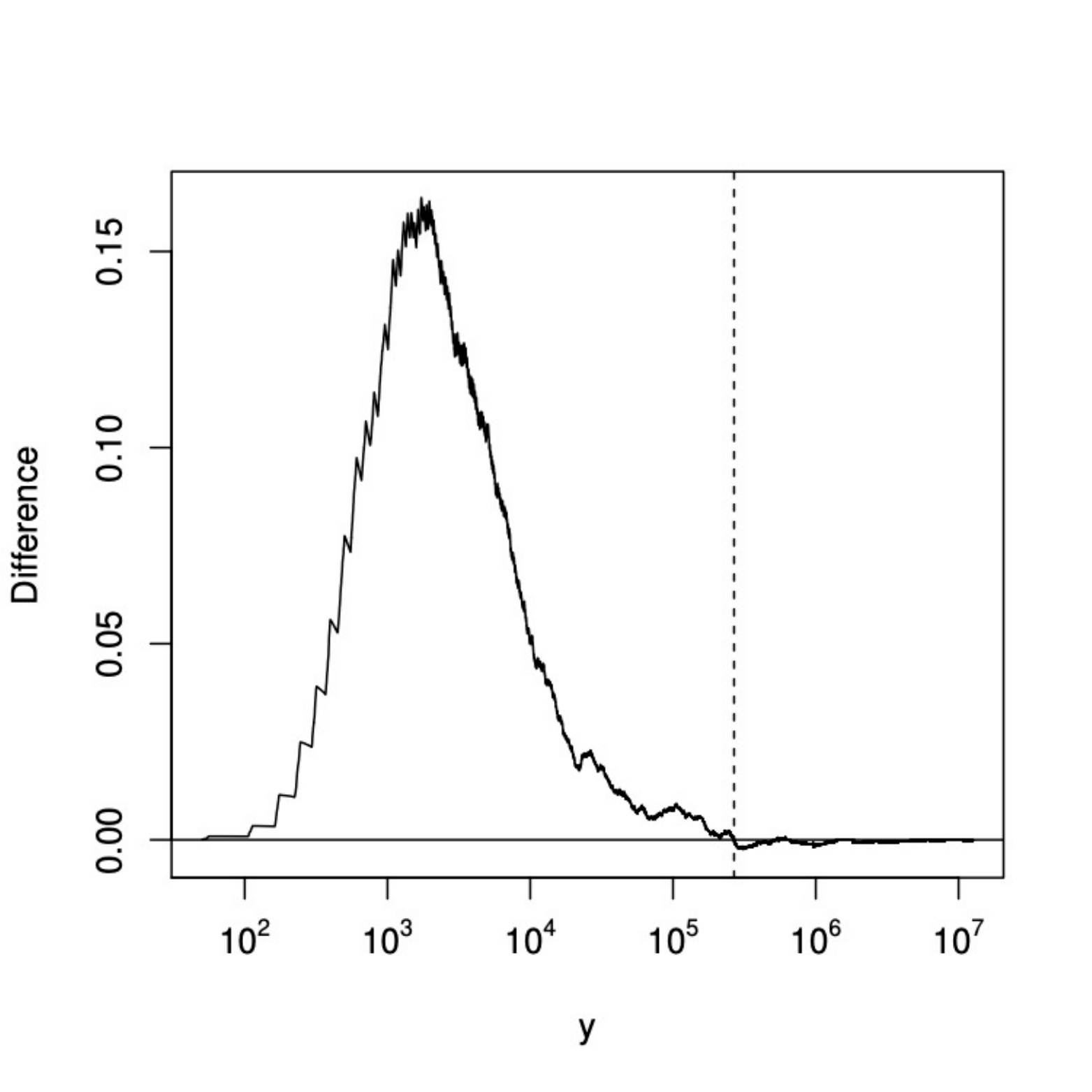}
        \subcaption{Differences of the distributions: $\widehat F_1 \oplus \widehat F_2-\widehat F_1 * \widehat F_2$}\label{fig:ecdfdiff}
    \end{minipage}
    \hfill
    \begin{minipage}[t]{.49\textwidth}
        \centering
        \includegraphics[width=\textwidth]{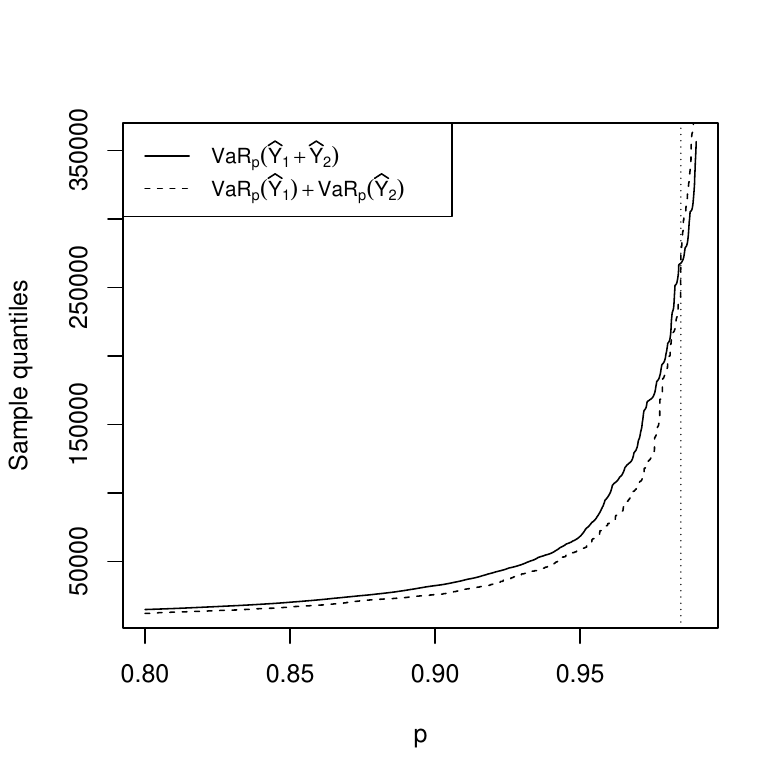}
        \subcaption{Sample quantiles for $p\in(0.8,0.99)$}\label{fig:qsamples}
    \end{minipage}
    \caption{Plots for $\widehat F_1 \oplus \widehat F_2-\widehat F_1 * \widehat F_2$ and sample quantiles}\label{fig:real-data-1}
\end{figure}

\subsection{Aggregation of Pareto risks with different parameters} 
\label{sec:63}
%   \begin{table}[htbp]
% \centering
% %\small
% \begin{tabular}{c|cccccc}
%       $i$  &1& 2 & 3& 4&5 &6 \\
%   \hline
%   Type&Cyclone& Flood & Flood, Storm& Storm&Bushfire &Hailstorm \\
%     \hline
%   Frequency &33& 25 & 27& 54&26 &33  \\
%    \hline
% \end{tabular}
% \caption{Types of catastrophic losses in Australia.}
% \label{t1}
% \end{table}
 As mentioned above, for independent losses $Y_1,\dots,Y_n$ following  generalized Pareto distributions with the same tail parameter $\alpha=1/\xi<1$,  it holds that  
\begin{align}\sum_{i=1}^n\VaR_{p}(Y_{i})\le\VaR_{p}\(\sum_{i=1}^n Y_{i}\), \mbox{~usually with strict inequality}.\label{eq:management}
\end{align} Inspired by the results in Section \ref{sec:62},
we are interested in whether  \eqref{eq:management}  holds for losses following generalized Pareto distributions with different parameters.
To make a first attempt on this problem, we look at the 6 operational losses of different business lines with infinite mean in Table 5 of \cite{moscadelli2004modelling}, where the operational losses are assumed to follow generalized Pareto distributions.
 Denote by $Y_1,\dots,Y_6$ the operational losses corresponding to these 6 generalized Pareto distributions. The estimated parameters in \cite{moscadelli2004modelling} for these losses are presented in Table \ref{t2}; they all have infinite mean.
  \begin{table}[htbp]
\centering
%\small
\begin{tabular}{c|cccccc}
      $i$  &1& 2 & 3& 4&5 &6 \\
  \hline
  $\xi_i$&1.19& 1.17 & 1.01& 1.39&1.23 &1.22 \\
    \hline
  $\beta_i$ &774& 254 & 233& 412&107 &243  \\
   \hline
\end{tabular}
\caption{The estimated parameters $\xi_i$ and $\beta_i$, $i\in[6]$.}
\label{t2}
\end{table}

For the purpose of this numerical example, we assume that $Y_1,\dots,Y_6$ are independent and plot $\sum_{i=1}^6\VaR_{p}(Y_{i})$ and $\VaR_{p}(\sum_{i=1}^6 Y_{i})$ for $p \in (0.95,0.99)$  in Figure \ref{f2}. We can see that  $\VaR_{p}(\sum_{i=1}^6Y_{i})$ is larger than $\sum_{i=1}^6\VaR_{p}(Y_{i})$,  and the gap between the two values gets larger as  the level $p$ approaches 1. This observation further suggests that even if the extremely heavy-tailed Pareto losses have different tail parameters, a diversification penalty may still exist. We conjecture that this is true for any generalized Pareto losses $Y_1,\dots,Y_n$ with shape parameters $\xi_1,\dots,\xi_n\in[1,\infty)$, although we do not have a proof. Similarly, we may expect that $\sum_{i=1}^n\theta_i\VaR_{p}(X_{i})\le\VaR_{p}(\sum_{i=1}^n \theta_i X_{i})$ holds for any Pareto losses $X_1,\dots,X_n$ with tail parameters $\alpha_1,\dots,\alpha_n\in(0,1]$,
 \begin{figure}[h]
\centering
\includegraphics[height=7cm, trim={0 0 0 20},clip]{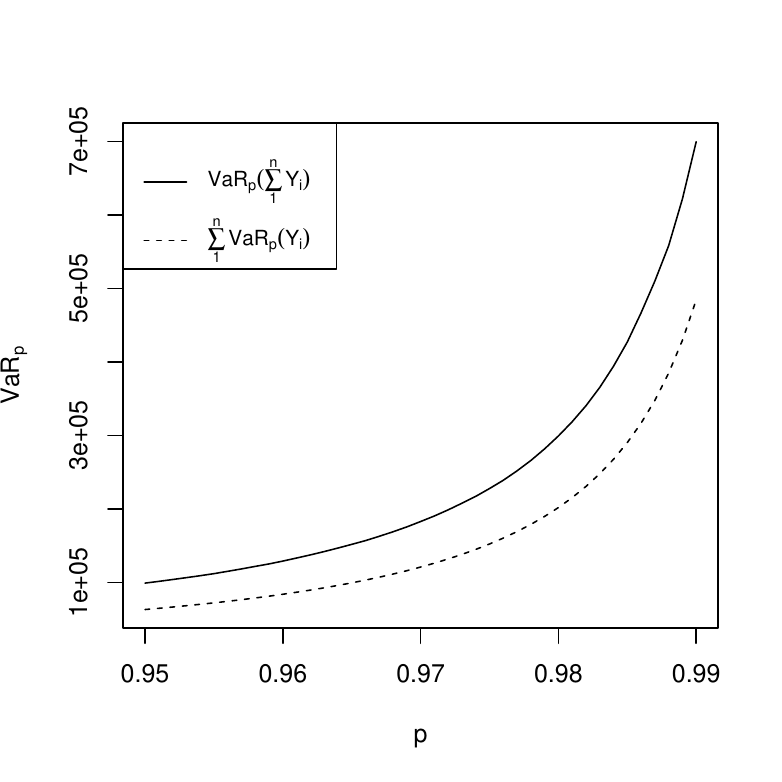}
\caption{Curves of $\VaR_{p}(\sum_{i=1}^nY_{i})$ and $\sum_{i=1}^n\VaR_{p}(Y_{i})$ for $n=6$ generalized Pareto losses with parameters in Table \ref{t2} and $p\in(0.95,0.99)$.}
\label{f2}
\end{figure}

From a risk management point of view, the message from Sections \ref{sec:62} and \ref{sec:63} is clear. If a careful statistical analysis leads to statistical models in the realm of infinite means, then the risk manager at the helm should take a step back and question to what extent classical diversification arguments can be applied. Though we mathematically analyzed the case of identically distributed losses, we conjecture that these results hold more widely in the heterogeneous case. As a consequence, it is advised to hold on to only one such super-Pareto risk. Of course, the discussion concerning the practical relevance of infinite mean models remains. When such underlying models are methodologically possible, then one should think carefully about the applicability of standard risk management arguments; this brings us back to Weitzman's Dismal Theorem as discussed towards the end of Section \ref{sec:1}. From a methodological point of view, we expect that the results from Sections \ref{sec:4} and \ref{sec:5} carry over to the above heterogeneous setting.

% The managerial insight from Sections \ref{sec:62} and \ref{sec:63} is very simple yet powerful: if data are approximately Pareto with infinite mean, then we expect the inequality \eqref{eq:management}  
% even if their parameters are different.  Therefore, it is better to hold on to one of the loss instead of diversifying across independent sources of such losses, even if such a diversification is possible. We expect all arguments in Sections \ref{sec:4} and \ref{sec:5} to carry through in this heterogeneous setting, and this seems to be true at least empirically. 

\section{Concluding remarks}\label{sec:7}

%Our main result (Lemma \ref{thm:1}) establishes   that a weighted average of  WNAID super-Pareto random variables, possibly triggered by  different events, is larger than one such loss  in the  sense of  first-order stochastic dominance. 
%The class of super-Pareto distributions includes many distributions that are more heavy-tailed than Pareto$(1)$ distribution, such as extremely heavy-tailed Pareto distributions.
%Our results provide an important implication in risk management, answering the question in the Introduction. 

We provide several generalizations of the inequality that the diversification of WNAID super-Pareto losses is greater than an individual super-Pareto loss in the sense of first-order stochastic dominance. The generalizations concern marginal distributions (Proposition \ref{prop:convolution}), dependence structures (Proposition \ref{prop:mixture}), a tail risk model (Proposition \ref{prop:tail}), a classic insurance model (Proposition \ref{cor:random}), and bounded super-Pareto losses (Theorem \ref{prop:bounded}). These results strengthen the main point made by \cite{CEW24}: As diversification increases the risk assessment of extremely heavy-tailed losses for all commonly used decision models, non-diversification is preferred.

% Next, we answer a question raised in the Introduction: If an agent aims to minimize their risk by allocating their exposures over extremely heavy-tailed losses, should they diversify or not? For a portfolio of super-Pareto losses, as diversification increases its risk assessment uniformly for all commonly used risk preferences, non-diversification is preferred.

The equilibrium of a risk exchange model is analyzed, where agents can take extra super-Pareto losses with compensations. In particular, if every agent is associated with an initial position of a super-Pareto loss, the agents can merely exchange their entire position with each other (Theorem \ref{th:equil}).  On the other hand, if some external agents are not associated with any initial losses, it is possible that all agents can reduce their risks by transferring the losses from the agents with initial losses to those without initial losses (Theorem \ref{th:external}).

  Inspired by numerical results,  an open question arises, 
%   The first question is whether 
% \begin{equation}\label{eq:q1}
% \frac{1}{k}\sum_{i=1}^kX_i\le_{\rm st}\frac{1}{\ell}\sum_{i=1}^\ell X_i,
% \end{equation}holds for all $k,\ell \in \mathbb N$ such that $k\le \ell$, where $X_{1},\dots,X_{l}$ are iid super-Pareto losses. The statement is true if $\ell$ is a multiple of $k$, as shown in Proposition \ref{prop:1}. 
that is whether 
\begin{align}\label{eq:q2}
\VaR_{p}\( \sum_{i=1}^n \theta_{i}X_{i}\)\ge \sum_{i=1}^n \theta_{i}\VaR_{p}(X_{i})
\end{align}
holds for $(\theta_1,\dots,\theta_n)\in\Delta_n$ and independent extremely heavy-tailed Pareto losses $X_1,\dots,X_n$ with possibly different tail parameters.
From the numerical results in  Section \ref{sec:6}, 
 \eqref{eq:q2} is anticipated  to hold; a proof seems to be beyond the current techniques.

\subsection*{Acknowledgements}  
The authors thank  Hansj\"org Albrecher, Jan Dhaene, John Ery, Taizhong Hu, Massimo Marinacci,  Alexander Schied, and Qihe Tang for helpful comments on a previous manuscript, which evolved into two papers (\cite{CEW24} and this paper).
RW is supported by the Natural Sciences and Engineering Research Council of Canada (RGPIN-2018-03823 and CRC-2022-00141).

\appendix

\setcounter{table}{0}
\setcounter{figure}{0}
\setcounter{equation}{0}
\renewcommand{\thetable}{A.\arabic{table}}
\renewcommand{\thefigure}{A.\arabic{figure}}
\renewcommand{\theequation}{A.\arabic{equation}}

\setcounter{theorem}{0}
\setcounter{proposition}{0}
\renewcommand{\thetheorem}{A.\arabic{theorem}}
\renewcommand{\theproposition}{A.\arabic{proposition}}
\setcounter{lemma}{0}
\renewcommand{\thelemma}{A.\arabic{lemma}}

\setcounter{example}{0}
\renewcommand{\theexample}{A.\arabic{example}}

\setcounter{corollary}{0}
\renewcommand{\thecorollary}{A.\arabic{corollary}}

\setcounter{remark}{0}
\renewcommand{\theremark}{A.\arabic{remark}}
\setcounter{definition}{0}
\renewcommand{\thedefinition}{A.\arabic{definition}}

% \begin{center}
% \Large   Appendices
% \end{center}

\section{Background on risk measures}\label{app:A} 
Recall that $\mathcal X_\rho$ is a convex cone of random variables representing losses
faced by financial institutions. We first present commonly used properties of a risk measure $\rho:\mathcal X_\rho \rightarrow \R$: 
 \begin{enumerate}[(a)]
 %\item Weak monotonicity: $\rho(X) \le \rho(Y)$ for $X,Y\in\mathcal X_\rho$ if $X\le_{\rm st} Y$. 
% \item Mild monotonicity: $\rho$ is weakly monotone and 
%$\rho(X)< \rho(Y)$ if $F^{-1}_X< F^{-1}_Y$ on $(0,1)$.
 \item[(d)] Translation invariance: $\rho(X+c)=\rho(X)+c$ for $c\in \R$.
 \item[(e)] Positive homogeneity: $\rho(aX)=a\rho(X)$ for $a\ge 0$.
 \item[(f)] Convexity: $\rho(\lambda X+(1-\lambda)Y)\le \lambda \rho(X)+(1-\lambda)\rho(Y)$ for $X,Y\in\mathcal X_\rho$ and $\lambda\in[0,1]$.
 %\item Law-invariance: $\rho(X)=\rho(Y)$ if $F_X=F_Y$.
\end{enumerate}
%Risk measures satisfying (a) and (b) are called monetary risk measures. In this paper, we also consider a slightly different version of monotonicity. For a risk measure $\rho:\mathcal X_\rho \rightarrow \R$ and $X,Y\in\mathcal X_\rho$,
% \begin{enumerate}[(a*)]
% \item Mild monotonicity: $\rho(X) \le \rho(Y)$ if $X\le_{\rm st} Y$ and 
%$\rho(X)< \rho(Y)$ if $F^{-1}_X(p) < F^{-1}_Y(p)$ for all $p\in(0,1)$.
%\end{enumerate}
 
 A risk measure that satisfies (b) weak monotonicity, (d) translation invariance, and (f) convexity is a \emph{convex risk measure} \citep{follmer2002convex}. ES is a convex risk measure.
 % For two random variables $X$ and $Y$ with finite mean and $\lambda\in(0,1)$, a convex risk measure $\rho$ satisfies the convexity property, i.e.,
%\begin{equation}\label{convexity}
%\rho(\lambda X+(1-\lambda)Y)\le \lambda \rho(X)+(1-\lambda)\rho(Y).
%\end{equation}
The convexity property means that diversification will not increase the risk of the  loss portfolio, i.e., the risk of $\lambda X+(1-\lambda)Y$ is less than or equal to that of  the weighted average of individual losses. However, the canonical space for law-invariant convex risk measures is $L^1$ (see \cite{filipovic2012canonical}) and hence convex risk measures are not useful for losses without finite mean.

For losses without finite mean, it is natural to consider VaR or Range Value-at-Risk (RVaR), which includes VaR as a limiting case.
 For $X\in \mathcal{X}$ and $0\leq p<q<1$,  RVaR is defined as
 $$\RVaR_{p,q}(X)=\frac{1}{q-p}\int_{p}^{q}\VaR_{u}(X)\d u.$$
 For $p\in(0,1)$, $\lim_{q\downarrow p^{+}}\RVaR_{p,q}(X)=\VaR_{p}(X)$. The class of RVaR is proposed by \cite{cont2010robustness} as robust risk measures; see \cite{ELW18} for its properties and risk sharing results.
% For generality, we consider the class of \emph{distortion risk measures}, defined as  
%\begin{align}\label{eq:distor}
%\rho_h(X)=\int_{-\infty}^0 (h(1-F_X(x))-1)\mathrm d x+\int_0^{+\infty}h(1-F_X(x))\mathrm{d}x,
%\end{align}
%where $X\in \mathcal X_\rho$ and $h:[0,1]\rightarrow[0,1]$ is a nondecreasing function with $h(0)=0$ and $h(1)=1$; $h$ is called the \emph{distortion function}.
 VaR, ES, RVaR, essential infimum ($\essinf$), and essential supremum ($\esssup$), belong to the family of distortion risk measures defined by \eqref{eq:distor}. For $X\in \mathcal X$, $\essinf$ and $\esssup$ are defined as
$$\essinf(X)=\sup\{x: F_X(x)=0 \}~~~~\mbox{and}~~~~ \esssup(X)=\inf\{x: F_X(x)=1 \}.$$ The distortion functions of $\essinf$ and $\esssup$ are  $h(t)=\id_{\{t=1\}}$ and $h(t)=\id_{\{0<t\le1\}}$, $t\in[0,1]$, respectively; see Table 1 of \cite{wang2020distortion}.  
 Distortion risk measures satisfy (b), (d) and (e). Almost all useful distortion risk measures are mildly monotone, as shown by the following result.
 
 \begin{proposition}\label{prop:mildly}
 Any distortion risk measure is mildly monotone unless it is  a mixture of $\esssup$ and $\essinf$.
 \end{proposition}
 
\begin{proof}
 Let $\rho_h$ be a distortion risk measure with distortion function $h$. 
 %We will prove an equivalent statement that if $\rho$ is not mildly monotone, then $h(t)=\lambda \id_{\{t=1\}}+(1-\lambda)\id_{\{0<t\le1\}}$, $t\in[0,1]$, for some $\lambda\in[0,1]$. 
 Suppose that $\rho_h$ is not mildly monotone. Then there exist $X,Y\in\mathcal X$ satisfying $F^{-1}_X(p)<F^{-1}_Y(p)$ for all $p\in(0,1)$ and $\rho(X)=\rho(Y)$.
 %Since the quantile function is left-continuous, 
 %we have $F^{-1}_X(p)<F^{-1}_Y(p)$. 
  % By the definition of distortion risk measures in \eqref{eq:distor},  
% $$0=\rho(X)-\rho(Y) = \int_{-\infty}^\infty h(1-F_X(x))-h(1-F_Y(x))\d x.$$
% Since $h(1-F_X(x))-h(1-F_Y(x))\le 0$ for all $x\in \R$,
 %we have $h(1-F_X(x)) = h(1-F_Y(x)) $ for almost every $x\in \R$.
 Suppose that there exist $b\in (0,1)$ such that $h(1-a)<h(1-b)$ for all $a> b$. 
For $x \in (F^{-1}_X(b) , F^{-1}_Y(b))$, we have 
 $F_X(x) \ge  b>F_Y(x)$; see e.g., Lemma 1 of \cite{GJW22}.
 Hence, we have $h(1-F_X(x)) \le h(1-b) < h(1-F_Y(x))  $ for $x \in (F^{-1}_X(b) , F^{-1}_Y(b))$.
 Since $h(1-F_X(x))-h(1-F_Y(x))\le 0$ for all $x\in \R$,
by \eqref{eq:distor} we get 
$$
\rho(X)-\rho(Y)= \int_{-\infty}^\infty \left( h(1-F_X(x))-h(1-F_Y(x)) \right) \d x <0.
$$ 
This contradicts $\rho(X)=\rho(Y)$. Hence, there is no $b\in (0,1)$ such that $h(1-a)<h(1-b)$ for all $a > b$. 
Using a similar argument with the left quantiles replaced by right quantiles, we conclude that there is no $b\in (0,1)$ such that $h(1-a)>h(1-b)$ for all $a < b$. 
Therefore, for every $b\in (0,1)$, there exists an open interval $I_b$ such that $b\in I_b$ and $h$ is  constant on $I_b$. 
For any $\epsilon>0$, the interval $[\epsilon, 1-\epsilon]$ is compact. Hence, there exists a finite collection $\{I_b:b\in B\}$ which covers $[\epsilon, 1-\epsilon]$. Since the open intervals in $\{I_b:b\in B\}$ overlap, we know that $h$ is  constant on $[\epsilon, 1-\epsilon]$. Letting $\epsilon \downarrow 0$ yields that $h$ takes a constant value on $(0,1)$, denoted by $\lambda\in [0,1]$. 
Together with $h(0)=0$ and $h(1)=1$, we get that  $h(t)= \lambda \id_{\{0<t\le1\}} + (1-\lambda) \id_{\{t=1\}}$ for $t\in[0,1]$, which is the distortion function of $\rho_h=\lambda \essinf + (1-\lambda) \esssup$.
 \end{proof}

  As a consequence,
  for any set $\X$ containing a random variable unbounded from above and one unbounded from below, such as the $L^q$-space for $q\in [0,\infty)$,  
   a real-valued distortion risk measure on $\X$ is  always mildly monotone.

\section{Proofs of all results}\label{app:proof}

\begin{proof}[Proof of Proposition \ref{prop:convolution}]
To show that $\mathcal F_{\rm IN}$
is  closed under convolution,  note that  first-order stochastic dominance is closed under convolution; see Theorem 1.A.3 of \cite{SS07}.
Therefore,  under independence, 
\begin{equation*}X_{1j}\le_{\rm st}\sum_{i=1}^n\theta_iX_{ij} \mbox{~for  $j=1,2$} ~\Longrightarrow~\sum_{j=1}^2 X_{1j}\le_{\rm st}\sum_{j=1}^2\sum_{i=1}^n\theta_iX_{ij}=\sum_{i=1}^n \theta_i\sum_{j=1}^2 X_{ij}.\end{equation*}  
To show that   $\mathcal F_{\rm IN}$ and $\mathcal F_{\rm WNA}$ are closed under strictly increasing convex transforms $f$, we note that 
if $Y\le_{\rm st} \sum_{i=1}^n\theta_iY_i$, then 
  $f( Y)\le_{\rm st}f(\sum_{i=1}^n\theta_iY_i)\le \sum_{i=1}^n\theta_{i}f(Y_{i}),$  where the first inequality follows since $\le_{\rm st}$ is preserved under increasing transforms, and the second inequality is due to convexity of $f$. Moreover,  
    strictly increasing transforms do not affect the dependence structure of $(Y_1,\dots,Y_n)$.
  %this is clear when $f$ is strictly increasing; see Lemma \ref{lem:WNApareto} when $f$ is not strictly increasing and the dependence is weak negative association. The case that the dependence is negative association is similar.
\end{proof}

\begin{proof}[Proof of Proposition \ref{prop:mixture}]
Copulas for independence and weak negative association are in $\mathcal C_{\mathrm{DP}}$ because of Theorem \ref{thm:1}.
The copula for comonotonicity is in $\mathcal C_{\mathrm{DP}}$  because $X_1= \sum_{i=1}^n\theta_i X_i$ almost surely  in case of comonotonicity.
Denote by $C$ a copula of  $(X_1,\dots,X_n)$. Let $X\laweq X_1$. Then,  there exists a random vector $(U_1,\dots,U_n)\sim C$ such that $(F_{X}^{-1}(U_1),\dots,F_{X}^{-1}(U_n))\laweq (X_1,\dots,X_n)$. Note that for $p\in(0,1)$, $\p(\sum_{i=1}^n\theta_{i}F_{X}^{-1}(U_i)\le p) $  is linear in the distribution of $(U_1,\dots,U_n)$.  Therefore, if   $\p(\sum_{i=1}^n\theta_{i}X_{i}\le p)\le \p(X\le p) $ for all $p\in(0,1)$ holds for two different copulas, it also holds for their mixtures. 
\end{proof}

% \begin{proof}[Proof of Proposition \ref{prop:1}]
% Let $Y_j= (\sum_{i=n(j-1)+1}^{jn}X_{i} )/n$, $j=1,\dots,m.$ By Theorem \ref{thm:1}, $X_{j}'\le_{\rm st} Y_j$ for $j=1,\dots,m,$ where $X_{1}',\dots,X_{m}'$ are iid super-Pareto. Note that $Y_{1},\dots,Y_{m}$ are also independent. As   first-order stochastic dominance is closed under convolutions (e.g., Theorem 1.A.3 (a) of \cite{SS07}), we obtain
% \begin{align*}
% X_{1}+\dots+X_{m}\laweq X_{1}'+\dots+X_{m}'\le_{\rm st}Y_{1}+\dots+Y_{m}&=\frac{X_{1}+\dots+X_{mn}}{n}.
% \end{align*}
% Dividing both sides by $m$ yields the desired inequality. 
% \end{proof}

%\subsection{Proofs of Propositions \ref{prop:tail} and  \ref{cor:random}}

\begin{proof}[Proof of Proposition \ref{prop:tail}]
% {\color{red}By a similar argument used in the proof of Lemma \ref{lem:WNApareto}, we can show that there exist WNAID super-Pareto random variables $X_{1},\dots,X_{n}$ and random variables $Y'_{1},\dots,Y'_{n}$ such that $(X_{1},\dots,X_{n})$ and $(Y'_{1},\dots,Y'_{n})$ have the same copula, and $(Y'_{1},\dots,Y'_{n})\laweq (Y_{1},\dots,Y_{n})$.}
%  % Let $ X_{1},\dots,X_{n}$ be iid \trd{super-Pareto} random variables. 
% \trd{As $ Y\ge_{\rm st} X $,     by  Theorem 6.B.14 in \cite{SS07} and Lemma \ref{thm:1}, for $t\ge x$,}
%  $$\p\(\sum_{i=1}^n\theta_{i}Y_{i}>  t\)=\p\(\sum_{i=1}^n\theta_{i}Y'_{i}>  t\) \ge \p\(\sum_{i=1}^n\theta_{i}X_{i}>  t\)\ge\p\(X> t\)  =  \p\(Y> t\).$$  
% The statement on strictness also follows from Lemma \ref{thm:1}.  
 Let $ X_{1},\dots,X_{n}$ be iid super-Pareto random variables. 
Note that  for $t\ge x$,   by using Theorem \ref{thm:1} and $ Y\ge_{\rm st} X $, we have 
 $$\p\(\sum_{i=1}^n\theta_{i}Y_{i}>  t\) \ge \p\(\sum_{i=1}^n\theta_{i}X_{i}>  t\)\ge\p\(X> t\)  =  \p\(Y> t\).$$  
The statement on strictness also follows from Theorem \ref{thm:1}.  
\end{proof}

\begin{proof}[Proof of Proposition \ref{cor:random}]
By Theorem \ref{thm:1}, it is clear that $\p(\sum_{i=1}^nW_iX_i/\sum_{i=1}^nW_i\le t)<\p(X \le t)$ for $t>z_X$, $n\in\N/\{1\}$. As $N$ is independent of $\{W_iX_i\}_{i\in \N},$  for $t> z_X$,
\begin{align*}
\p\(\frac{\sum_{i=1}^NW_iX_i}{\sum_{i=1}^N W_i}\le t\)&=\p(N=0)+\sum_{n=1}^\infty\p\(\frac{\sum_{i=1}^nW_iX_i}{\sum_{i=1}^n W_i}\le t\)\p(N=n)\\
&\le\p(N=0)+\p(N\ge1)\p(X\le t)=\p\(X\id_{\{N\ge1\}}\le t\).
\end{align*}
It is obvious that the inequality is strict if $\p(N\ge2)\neq0$.
To show the second inequality in \eqref{eq:collective},  note that for each realization of $N=n$ and $(W_1,\dots,W_N)=(w_1,\dots,w_n)\in \R^n$, 
$\sum_{i=1}^n w_i X  \le_{\rm st} {\sum_{i=1}^ n w_iX_i} $ holds by Theorem \ref{thm:1}. Hence, the second inequality in \eqref{eq:collective} holds. 
\end{proof}

 \begin{proof}[Proof of Theorem \ref{prop:bounded}]
 For $t\in(z_X,c]$,  we have
\begin{align*}
\p\(\sum_{i=1}^n\theta_{i}Y_{i}\le  t\)&=\p\(\sum_{i=1}^n\theta_{i}(X_i\wedge c_i)\le  t\) 
 %\\
 %&=\p\(\left\{\sum_{i=1}^n\theta_{i}(X_i\wedge c_i)\le  t\right\}\cap \bigcap_{i=1}^n\{X_i\le c_i\}\)\\
% &=\p\(\left\{\sum_{i=1}^n\theta_{i}X_i\le  t\right\}\cap \bigcap_{i=1}^n\{X_i\le c_i\}\)
=
\p\(\sum_{i=1}^n\theta_{i}X_{i}\le  t\).
\end{align*}
We also have $\p(Y_i\le t)=\p(X_i\wedge c_i\le t)=\p(X_i\le t)$, $i\in[n]$. By the strictness statement in Theorem \ref{thm:1}, we obtain the probability inequality. 
To show the quantile inequality,  note that for $p\in (0,\p(X\le c))$, 
we have $
\VaR_p(X ) < c.
$
Using Theorem \ref{thm:1} and the definition of $Y_1,\dots,Y_n$,  we get
\begin{align*}
 \sum_{i=1}^n \theta_{i} \VaR_p(Y_i)  \le \VaR_p(X ) & <    \VaR_p\left( \sum_{i=1}^n \theta_{i}X_{i}\right) \wedge c 
\\ &\le \VaR_p\left(   \left(    \sum_{i=1}^n \theta_{i}X_{i}\right)  \wedge  c  \right)
\le \VaR_p\left(    \sum_{i=1}^n \theta_{i}Y_{i}\right) .
\end{align*}
This gives the desired inequality. 
\end{proof}

%\subsection{Proofs of  Propositions \ref{prop:mildly}, \ref{prop:risk-meas} and   \ref{prop:concentrate}}

% \begin{proof}[Proof of Proposition \ref{prop:risk-meas}]
% This proposition  follows directly from  Theorem \ref{thm:1}.
% \end{proof}

% \begin{proof}[Proof of Proposition \ref{prop:concentrate}]
% The proposition follows directly from Theorem \ref{thm:1}.  
% \end{proof}

\begin{proof}[Proof of Theorem \ref{th:equil}]
\begin{enumerate}[(i)]
\item 
Suppose that $\(\mathbf p^*,\mathbf w^{1*},\dots,\mathbf w^{n*}\)$ forms an equilibrium.   
We let $p= \max_{j\in[n]}\{p_j\}$ and $S=\argmax_{j\in[n]}\{p_j\}$. 
For agent $i$,  
by writing   $w=\Vert \mathbf w^i\Vert $,  using Theorem \ref{thm:1}  and the fact that $\rho_i$ is mildly monotone, we have that for any $\mathbf w^i\in \R_+^n$,
\begin{align*}
\rho_i(L_{i}(\mathbf w^{i}, \mathbf p^*))&=
 \rho_i(\mathbf w^{i} \cdot(\mathbf X -\mathbf p^*) + \mathbf a^i\cdot\mathbf p^*) \\ &\ge  \rho_i(\mathbf w^{i} \cdot \mathbf X -  w p + \mathbf a^i\cdot\mathbf p^*)  \ge \rho_i(w  X_1 -  w p + \mathbf a^i\cdot\mathbf p^*).    \end{align*} 
 By the last statement of Theorem \ref{thm:1}, the last inequality is strict if $\mathbf w^i$ contains at least two non-zero components. 
 Moreover, $c(\Vert \mathbf w^i\Vert-\Vert \mathbf a^i\Vert)=c(w-\Vert \mathbf a^i\Vert)$.  
 Therefore, the optimizer $\mathbf w^{i*}=(w_1^{i*},\dots,w_n^{i*})$  to \eqref{eq:opt-internal} has at most one non-zero component $w_j^{i*}$ for  $j\in S$.
 Hence, $w_k^{i*} =0$ if  $k\in [n]\setminus S$ and this holds for each $i\in [n]$.  
Using $\sum_{i=1}^n \mathbf w^{i*}=\sum_{i=1}^n \mathbf a^i$ which have all positive components, we know that $S=[n]$, which further implies that $\mathbf p^*=\( p,\dots,p\)$ for $p\in \R_+$. Next, as each $\mathbf w^{i*}$ has only one positive component, $(\mathbf w^{1*},\dots,\mathbf w^{n*})$ has to be an $n$-permutation of $(\mathbf a^1,\dots,\mathbf a^n)$ in order to satisfy the clearance condition \eqref{eq:clearance-internal}. 

\item The clearance condition \eqref{eq:clearance-internal} is clearly satisfied. As distortion risk measures are translation invariant and positive homogeneous (see Appendix \ref{app:A}), by Proposition \ref{prop:concentrate},  
  for $i\in [n]$,
\begin{align}
&\min_{\mathbf w^i\in\mathbb R_+^n} \left\{ \rho_i\(L_{i}(\mathbf w^i, \mathbf p^*)\) + c_i(\Vert \mathbf w^i \Vert -\Vert \mathbf a^i \Vert)\right\} \notag \\
&=\min_{\mathbf w^i\in\mathbb R_+^n}\left\{\rho_i\(\mathbf w^i \cdot\mathbf X - (\mathbf w^i-\mathbf a^i)\cdot\mathbf p^*\)+ c_i(\Vert \mathbf w^i \Vert -\Vert \mathbf a^i \Vert)\right\} \notag \\
&= \min_{\|\mathbf w^i\|\in\mathbb R_+}\left\{(\rho_i\(\|\mathbf w^i\|  X\)-(\|\mathbf w^i\|-a_i)p)+ c_i(\Vert \mathbf w^i \Vert -\Vert \mathbf a^i \Vert)\right\}  \notag \\
&=\min_{ w\in\mathbb R_+} \left\{ w (\rho_i(X)-p)+a_ip+ c_i(w-a_i)\right\} .\label{eq:optimized} 
\end{align} 
Note that  $w\mapsto w (\rho_i(X)-p)+ c_i(w-a_i)$ is convex and  with condition \eqref{eq:convex}, its minimum is attained at $w=a_i$. 
Therefore, $\mathbf w^{i*}=\mathbf a^{i*}$ is an optimizer to \eqref{eq:opt-internal}, which shows the desired equilibrium statement.

\item
By (i), $\(\mathbf w^{1*},\dots,\mathbf w^{n*}\)$ is an $n$-permutation of $(\mathbf a^1,\dots,\mathbf a^n)$. It means that for any $i\in[n]$, there exists $j\in[n]$ such that $a_j$ is the minimizer of \eqref{eq:optimized}. As $c_i$ is convex, we have
$$c_{i+}'(a_j-a_i) \ge  p-\rho_i(X)  \ge  c_{i-}'(a_j-a_i),~~~~\mbox{for each}~~~~i\in[n].$$
Hence, we obtain \eqref{eq:convex2}. \qedhere
\end{enumerate}
\end{proof}

\begin{proof}[Proof of Theorem \ref{th:external}]
As in Section \ref{sec:52}, an optimal position for either the internal or the external agents is to concentrate on one of the losses $X_i$, $i\in[n]$. 
By the same arguments as in Theorem \ref{th:equil} (i), the equilibrium price, if it exists, must be of the form $\mathbf p=(p,\dots,p)$.
For such a given $\mathbf p$, 
using the assumption that $\rho_{E}$ and $\rho_I$ are mildly monotone and    Proposition \ref{prop:concentrate}, we can rewrite the optimization problems in \eqref{eq:opt-external1}
 and  \eqref{eq:opt-external2} as 
\begin{align}\label{eq:opt-external1-p}
 {\min_{\mathbf u^j\in\mathbb R_+^n}}\left\{  \rho_{E}\(L_{E}(\mathbf u^j, \mathbf p)    \)+ c_{E}( \|\mathbf u^j\| )\right\}
&={\min_{ u\in\mathbb R_+}}\left\{ u\(\rho_{E}\(X\)-p\)+c_{E}( u )\right\},
\end{align}
and
\begin{align}
\label{eq:opt-external2-p}
{\min_{\mathbf w^i\in\mathbb R_+^n}}\left\{ \rho_I\(L_{i}(\mathbf w^i, \mathbf p)\)+c_{I}( \|\mathbf w^i\| -\|
  \mathbf a^i\|)\right\}={\min_{ w\in\mathbb R_+}}\left\{ w(\rho_I\(X\)-p)+ap+c_{I}( w-a)\right\},\end{align}
  for  $j\in[m]$  and $i\in[n]$.
  Note that the derivative of the function inside the minimum of the right-hand side of  \eqref{eq:opt-external1-p}  with respect to $u$ 
  is $L_E(u)-p$,
  and similarly, $L_I(w-a)-p$ is  the derivative of the function inside the minimum of the right-hand side of  \eqref{eq:opt-external2-p}.
    Using strict convexity of $c_E$ and $c_I$, we get the following facts.
\begin{enumerate}
\item  The optimizer $u$ to \eqref{eq:opt-external1-p} has two cases: \begin{enumerate}
  \item If $L_E(0)\ge p$, then  $u=0$.
\item  If   $ L_E(0)<p$, then  $u>0$ and  
 $L_E(u)= p $. 
   \end{enumerate}
  \item 
    The optimizer $w$ to \eqref{eq:opt-external2-p} has four cases: \begin{enumerate}
    \item If 
  $L_{I}^+(0) < p$, then $w>a$.    This is not possible in an equilibrium.
\item  If   $L_{I}^+(0) \ge  p \ge  L_{I}^-(0) $, then $w=a$.
\item If   $ L_{I}^-(0)>  p >  L_{I}(-a) $, then $0 < w<a$ and 
$L_I(w-a)=p$. 
\item If   $ L_{I}(-a) \ge   p  $, then $w=0$.    \end{enumerate}
  \end{enumerate}

From the above analysis, we see that the optimal positions for the external agents are either all $0$ or all positive, and they are identical due to the strict monotonicity of $L_E$. We can say the same for the internal agents.  Suppose that there is an equilibrium. Let $u$ be the external agent's common exposure, and $w$ be the internal agent's exposure.  By the clearance condition \eqref{eq:clearance-external} we have $w+ku=a$.   
If $0<ku<a$, then from (1.b) and (2.c) above, we have $L_E(u)=L_I(-ku)$.  Below we show the three statements. 
%Therefore, since each loss has the same total amount $a$,  
%in an equilibrium,
%we must have precisely $1$ internal and $k$ external agents sharing a loss $X_i$ for each $i\in [n]$.
%We can reduce our problem to the sub-problem of sharing one loss $X$.  

%Suppose that in this sub-problem, there is an equilibrium. Let $u$ be the external agents' common exposure. Since the total exposure is $a$, we know that the internal agent's exposure is $w=a-ku$.   
%If $0<ku<a$, then from (1.b) and (2.c) above, we have $L_E(u)=L_I(-ku)$.  Below we show the three statements. 

 \begin{enumerate}[(i)] 
\item 
%If $L_E(a/k)<L_I(-a)$, then by strict monotonicity of $L_E$ and $L_I$, there is no $u \in (0,a/k]$ such that $L_E(u)=L_I(-ku)$. 
%Since $u$ cannot be larger than $a/k$, 
%if an equilibrium exists, then $u=0$; but in this case, by (1.a) and (2.b), we have $L_E(0)\ge p\ge L_{I-}(0)$, which contradicts $L_E(a/k)<L_I(-a)$. 
%Hence, there is no equilibrium. 

The ``if" statement is clear by (1.b) and (2.d). We show the ``only if" statement. We first assume that $p\in(L_I(-a),L_I^-(0))$. Since we also have  $p>L_E(a/k)>L_E(0)$, from (1.b) and (2.c), $u$ should satisfy $L_E(u)=L_I(-ku)$. However, by strict monotonicity of $L_E$ and $L_I$, there is no $u \in (0,a/k]$ such that $L_E(u)=L_I(-ku)$. Moreover, if $p\le L_E(0)$ or $p\ge L_I^-(0)$, the clearance condition \eqref{eq:clearance-external} cannot be satisfied. Therefore, we must have  $L_E(0)<p\le L_I(-a)$.  In this case, $w=0$ by (2.d). Consequently, $u=a/k$ by the clearance condition \eqref{eq:clearance-external} and (1.b) gives $p=L_E(a/k)$.
\item In this case,  there exists a unique $u^*\in (0,a/k]$ such that 
$L_E(u^*)=L_I(-ku^*)$. It follows that $u=u^*$ optimizes \eqref{eq:opt-external1-p} and $w=a-ku^*$ optimizes \eqref{eq:opt-external2-p}.
It is straightforward to verify that  $\mathcal E$ is an equilibrium, and thus the ``if" statement holds.
To show the ``only if" statement, it suffices to notice that 
$L_E(u)=L_I(-ku)=p$ has to hold, where $p$ is an equilibrium price and $u$ is the optimizer to \eqref{eq:opt-external1-p}, and such $u$ and $p$ are unique. 
Next, we show the ``only if'' statement for $u^*< a/2k$. As the optimal position for each external agent is  $a-ku^*>a/2$, if more than two internal agents take the same loss, then the clearance condition \eqref{eq:clearance-external} does not hold. Hence, the internal agents have to take different losses. Moreover, as the optimal position for the internal agents are the same, the loss $X_i$ for each $i\in [n]$, must be shared by one internal and $k$ external agents. The equilibrium is preserved under permutations of allocations. Thus, we have the ``only if '' statement for $u^*< a/2k$. The ``if'' statement is obvious.
 
\item The ``if" statement can be verified  directly by  using Theorem \ref{th:equil} (ii). Next, we show the ``only if" statement. By (2.a), it is clear that the equilibrium price $p$ satisfies $p\le L_{I}^+(0)$. If $p<L_I^-(0)$, by (1.a), (2.c), and (2.d), the clearance condition \eqref{eq:clearance-external} cannot be satisfied. Thus, $p\ge L_I^-(0)$. By a similar argument, we have $p\le L_E(0)$. Hence, we get $p\in [L_{I}^-(0),  L_{E}(0)\wedge L_{I}^+(0)]$. From (1.a) and (2.b), we have $u=0$ and $w=a$ and thus the desired result.  \qedhere
  \end{enumerate}
\end{proof}

\begin{proof}[Proof of Proposition \ref{prop:equil-ES}]
The clearance condition \eqref{eq:clearance-internal} is  clearly satisfied. As ES is translation invariant, it suffices to show that $\mathbf w^{i*}$ minimizes $\ES_q(\mathbf w^{i}\cdot\mathbf X-\mathbf w^{i}\cdot\mathbf p^*) + c_i( \Vert \mathbf w^i \Vert -  \Vert \mathbf a^i\Vert)$ for $i\in[n]$. 
Write $r: \mathbf w\mapsto  \ES_q\(\mathbf w \cdot \mathbf X\)$ for  $\mathbf w=(w_1,\dots,w_n)\in [0,1]^n$.
By Corollary 4.2 of \cite{tasche2000conditional}, 
\begin{align*}
\frac{\partial 
r 
}{\partial w_i} \(\mathbf w\) &=\E\[X_i|A_{\mathbf w}\], \mbox{~~~}i\in[n],
\end{align*}
where $A_{\mathbf w} =\{\sum_{i=1}^n w_iX_i\ge \VaR_q\(\sum_{i=1}^n w_iX_i\)\}$.
Moreover, using convexity of $r$, we have (see \citet[p.~321]{MFE15})
\begin{align*}
r\( \mathbf w \)-  \mathbf w\cdot\mathbf p^* \ge\sum_{i=1}^{n} w_i\frac{\partial r  }{\partial w_i} (a_1,\dots,a_n ) -  \mathbf w\cdot\mathbf p^* = 0.
\end{align*} 
By Euler's rule (see  \citet[(8.61)]{MFE15}), the equality holds if $\mathbf w= \lambda (a_1,\dots,a_n)$  for any $\lambda>0$ .
By taking $\lambda=a_i /\sum_{j=1}^n a_j$, we get 
$\Vert \mathbf w \Vert =a_i= \Vert \mathbf a^i\Vert$, and hence $c_i ( \Vert \mathbf w \Vert -  \Vert \mathbf a^i\Vert)$ is minimized by $\mathbf w=\lambda (a_1,\dots,a_n)$. 
Therefore, $
\mathbf w^{i*}$ 
is an optimizer for each $i\in [n]$. 
%The uniqueness statement follows if $c_i$ is non-zero except at $0$ and $r$ is strictly convex but we did not show it.
\end{proof}

{

}


\begin{thebibliography}{10}

%\bibitem[Artzner et~al., 1999]{artzner1999coherent}
%Artzner, P., Delbaen, F., Eber, J.-M. and Heath, D. (1999).
%\newblock Coherent measures of risk.
%\newblock {\em Mathematical Finance}, 9(3):203--228.


%\bibitem[\protect\citeauthoryear{Chen et al.}{Chen et al.}{2021}]{CLTW21}
%Chen, Y., Liu, P. and Tan, K. S. and Wang, R. (2021). Trade-off between validity and efficiency of merging p-values under arbitrary dependence. \emph{Statistica Sinica}, forthcoming.

\bibitem[\protect\citeauthoryear{Alam and  Saxena}{1981}]{AS81}
Alam, K. and Saxena, K. M. L. (1981). Positive dependence in multivariate distributions. \emph{Communications in Statistics-Theory and Methods}, {10}(12):1183--1196.

\bibitem[Alink et~al., 2004]{alink2004diversification}
Alink, S., L{\"o}we, M. and W{\"u}thrich, M.~V. (2004).
\newblock Diversification of aggregate dependent risks.
\newblock {\em Insurance: Mathematics and Economics}, 35(1):77--95.

%\bibitem[Albrecher et al., 2006]{AAK06}
%Albrecher, H., Asmussen, S. and Kortschak, D. (2006). 
%\newblock 
%Tail asymptotics for the sum of two heavy-tailed dependent risks. 
%\newblock {\em  Extremes}, 9(2):107--130.

%\bibitem[Andriani and McKelvey, 2007]{andriani2007beyond}
%Andriani, P. and McKelvey, B. (2007).
%\newblock Beyond Gaussian averages: Redirecting international business and
%  management research toward extreme events and power laws.
%\newblock {\em Journal of International Business Studies}, 38(7):1212--1230.

 \bibitem[Arab et~al., 2024]{ALO24}
Arab, I., Lando, T., and Oliveira, P.~E. (2024).
\newblock Convex combinations of random variables stochastically dominate the
  parent for a large class of heavy-tailed distributions.
\newblock {\em arXiv:2411.14926}.


\bibitem[\protect\citeauthoryear{Artzner et al.}{Artzner et al.}{1999}]{ADEH99}
{Artzner, P., Delbaen, F., Eber, J.-M. and Heath, D.} (1999). Coherent measures of risk. \emph{Mathematical Finance}, {9}(3):203--228.


\bibitem[Axtell, 2001]{axtell2001zipf}
Axtell, R.~L. (2001).
\newblock Zipf distribution of U.S. firm sizes.
\newblock {\em Science}, 293:1818--1820.


\bibitem[Balkema and de Haan, 1974]{BD74}
Balkema, A. and de Haan, L. (1974). 
\newblock Residual life time at great age.
\newblock {\em Annals of Probability}, 2:792--804.

\bibitem[\protect\citeauthoryear{Barrett and Donald}{Barrett and
  Donald}{2003}]{barrett2003consistent}
Barrett, G.~F. and Donald, S.~G. (2003).
\newblock Consistent tests for stochastic dominance.
\newblock {\em Econometrica}, 71(1):71--104.



%\bibitem[\protect\citeauthoryear{Baucells and Heukamp}{2006}]{BH06}  Baucells, M. and Heukamp, F. H. (2006). Stochastic dominance and cumulative prospect theory. \emph{Management Science},  {52}(9):1409--1423.

%\bibitem[Bauer and Zanjani, 2016]{bauer2016marginal}
%Bauer, D. and Zanjani, G. (2016).
%\newblock The marginal cost of risk, risk measures, and capital allocation.
%\newblock {\em Management Science}, 62(5):1431--1457.
 
 
 
\bibitem[Beirlant et~al., 1999]{beirlant1999tail}
Beirlant, J., Dierckx, G., Goegebeur, Y. and Matthys, G. (1999).
\newblock Tail index estimation and an exponential regression model.
\newblock {\em Extremes}, 2(2):177--200.


\bibitem[Biffis and Chavez, 2014]{BC14}
Biffis, E. and Chavez, E. (2014).
\newblock Tail risk in commercial property insurance.
\newblock {\em Risks}, 2(4):393--410.

  \bibitem[Block et~al., 1982]{block1982some}
Block, H.~W., Savits, T.~H. and Shaked, M. (1982).
\newblock Some concepts of negative dependence.
\newblock {\em Annals of Probability}, 10(3):765--772.

\bibitem[Block et~al., 1985]{block1985concept}
Block, H.~W., Savits, T.~H. and Shaked, M. (1985).
\newblock A concept of negative dependence using stochastic ordering.
\newblock {\em Statistics and Probability Letters}, 3(2):81--86.

%\bibitem[Caball\'e and Pomansky, 1996]{CP96}
%Caball\'e, J. and Pomansky, A. (1996). Mixed risk aversion. \emph{Journal of Economic Theory}, 71(2):485--513.

\bibitem[\protect\citeauthoryear{Chen et al.}{2025}]{CEW24}
Chen, Y., Embrechts, P. and Wang, R. (2025). An unexpected stochastic dominance: Pareto distributions, dependence, and diversification. \emph{Operations Research}, 73(3):1151--1722.

\bibitem[Chen and Shneer, 2025]{CS24}
Chen, Y. and Shneer, S. (2025).
\newblock Risk aggregation and stochastic dominance for a class of heavy-tailed distributions.
\newblock {\em ASTIN Bulletin}, forthcoming.



\bibitem[\protect\citeauthoryear{Chen and Wang}{2025}]{CW25}
Chen, Y.  and Wang, R. (2025).  Infinite-mean models in risk management: Discussions and recent advances. \emph{Risk Sciences}, 1:100003.


\bibitem[Cheynel et~al., 2024]{cheynel2022fraud}
Cheynel, E., Cianciaruso, D. and Zhou, F. (2024).
\newblock Fraud power laws.
\emph{Journal of Accounting Research}, 62(3):833--876.

\bibitem[Cirillo and Taleb, 2020]{cirillo2020tail}
Cirillo, P. and Taleb, N.~N. (2020).
\newblock Tail risk of contagious diseases.
\newblock {\em Nature Physics}, 16(6):606--613.






\bibitem[Clark, 2013]{clark2013note}
Clark, D.~R. (2013).
\newblock A note on the upper-truncated Pareto distribution.
\newblock  {\em Casualty Actuarial Society E-Forum}, 
Winter, 2013, Volume 1, pp. 1--22.


\bibitem[\protect\citeauthoryear{Cont, Deguest and Scandolo}{Cont
  et~al.}{2010}]{cont2010robustness}
Cont, R., Deguest, R. and Scandolo, G. (2010).
\newblock Robustness and sensitivity analysis of risk measurement procedures.
\newblock {\em Quantitative Finance}, {10}(6):593--606.

  

%\bibitem[\protect\citeauthoryear{Chew et al.}{1987}]{CKS87}
%Chew, S. H., Karni, E. and Safra, Z. (1987). Risk aversion in the theory of expected utility with rank dependent probabilities. \emph{Journal of Economic Theory},  {42}:370--381.

% \bibitem[\protect\citeauthoryear{Cossette et al.}{2003}]{CDM03}
% Cossette, H., Duchesne, T. and Marceau, \'E. (2003). Modeling catastrophes and their impact on insurance portfolios. \emph{North American Actuarial Journal}, {7}(4):1--22.


%\bibitem[\protect\citeauthoryear{de Haan and Ferreira}{de Haan and Ferreira}{2006}]{DF06}
%{de Haan L. and Ferreira A.} (2006). \emph{Extreme Value Theory: An Introduction}. Springer.

% \bibitem[\protect\citeauthoryear{Dhaene et al.}{Dhaene et al.}{2006}]{DDGKTV06}
% {Dhaene, J., Vanduffel, S., Goovaerts, M.J., Kaas, R., Tang, Q. and Vynche, D.} (2006). {Risk measures and comonotonicity: A review}. \emph{Stochastic Models}, 22(4):573--606.

\bibitem[Eling and Schnell, 2020]{ES20}
Eling, M. and Schnell, W. (2020).
\newblock Capital requirements for cyber risk and cyber risk insurance: An
  analysis of Solvency II, the US risk-based capital standards, and the Swiss
  Solvency Test.
\newblock {\em North American Actuarial Journal}, 24(3):370--392.

\bibitem[Eling and Wirfs, 2019]{EW19}
Eling, M. and Wirfs, J. (2019).
\newblock What are the actual costs of cyber risk events?
\newblock {\em European Journal of Operational Research}, 272(3):1109--1119.


\bibitem[\protect\citeauthoryear{Embrechts et al.}{Embrechts et al.}{1997}]{EKM97}
Embrechts, P., Kl\"uppelberg, C. and Mikosch, T. (1997). \emph{Modelling Extremal Events for Insurance and Finance}. Springer, Heidelberg.

\bibitem[Embrechts et~al., 1999]{embrechts1999extreme}
Embrechts, P., Resnick, S.~I. and Samorodnitsky, G. (1999).
\newblock Extreme value theory as a risk management tool.
\newblock {\em North American Actuarial Journal}, 3(2):30--41.

%\bibitem[Embrechts et~al., 2002]{embrechts2002correlation}
%Embrechts, P., McNeil, A. and Straumann, D. (2002).
%\newblock Correlation and dependence in risk management: properties and
%  pitfalls.
%In \newblock {\em Risk Management: Value at Risk and Beyond} (Eds: Dempster) pp.~176--223. Cambridge University Press.



\bibitem[\protect\citeauthoryear{Embrechts et al.}{2009}]{ELW09}
Embrechts, P., Lambrigger, D. and W{\"u}thrich, M. (2009).
 Multivariate extremes and the aggregation of dependent risks: examples and counter-examples.
{\em Extremes}, {12}(2):107--127.


\bibitem[\protect\citeauthoryear{Embrechts et al.}{2018}]{ELW18}
Embrechts, P., Liu, H. and Wang, R. (2018). Quantile-based risk sharing. \emph{Operations Research}, 66(4):936--949.





%\bibitem[\protect\citeauthoryear{Embrechts and Puccetti}{2010}]{EP10}
%Embrechts, P. and  Puccetti, G. (2010). Risk aggregation. In \emph{Copula Theory and its Applications} (Eds: Jaworski et al.) pp.~111--126. Springer, Heidelberg.

%\bibitem[Fama and Miller, 1972]{FM72}
%Fama, E.~F. and Miller, M.~H. (1972).
%\newblock {\em The Theory of Finance}.
%\newblock Dryden Press, Hinsdale.

\bibitem[Filipovi{\'c} and Svindland, 2012]{filipovic2012canonical}
Filipovi{\'c}, D. and Svindland, G. (2012).
\newblock The canonical model space for law-invariant convex risk measures is
  $L^1$.
\newblock {\em Mathematical Finance}, 22(3):585--589.


\bibitem[FINMA, 2021]{FINMA2021}
FINMA (2021).
\newblock Standardmodell Versicherungen (Standard Model Insurance): Technical description for the SST standard model non-life insurance (in German), 
October 31, 2021,
 \url{www.finma.ch}.




 \bibitem[Flyvbjerg et~al., 2022]{flyvbjerg2022empirical}
Flyvbjerg, B., Budzier, A., Lee, J.~S., Keil, M., Lunn, D. and Bester, D.~W.
  (2022).
\newblock The empirical reality of IT project cost overruns: Discovering a
  power-law distribution.
\newblock {\em Journal of Management Information Systems}, 39(3):607--639.


 
\bibitem[F{\"o}llmer and Schied, 2002]{follmer2002convex}
F{\"o}llmer, H. and Schied, A. (2002).
\newblock Convex measures of risk and trading constraints.
\newblock {\em Finance and Stochastics}, 6(4):429--447.


 \bibitem[\protect\citeauthoryear{F\"ollmer and Schied}{F\"ollmer and Schied}{2016}]{FS16} F\"ollmer, H.~and Schied, A.~(2016). \emph{Stochastic Finance. An Introduction in Discrete Time}. Fourth Edition.  {Walter de Gruyter, Berlin}.
 


 


 
  \bibitem[\protect\citeauthoryear{Guan et al.}{Guan et al.}{2024}]{GJW22} 
 Guan, Y., Jiao, Z. and Wang, R. (2024). A reverse ES (CVaR) optimization formula.
 \newblock {\em North American Actuarial Journal}, 28(3):611-625.


 
 \bibitem[Gabaix, 1999]{gabaix1999zipf}
Gabaix, X. (1999).
\newblock Zipf's law and the growth of cities.
\newblock {\em American Economic Review}, 89(2):129--132.


%\bibitem[\protect\citeauthoryear{Gabaix}{Gabaix}{2009}]{G09}
%Gabaix, X. (2009). Power laws in economics and finance. \emph{Annual Review of Economics}, 1(1):255--294.

\bibitem[Gabaix and Ibragimov, 2011]{gabaix2011rank}
Gabaix, X. and Ibragimov, R. (2011).
\newblock Rank-1/2: A simple way to improve the OLS estimation of tail
  exponents.
\newblock {\em Journal of Business and Economic Statistics}, 29(1):24--39.


 %\bibitem[\protect\citeauthoryear{F\"ollmer and Schied}{F\"ollmer and Schied}{2016}]{FS16} F{\"o}llmer, H. and Schied, A. (2016). \emph{Stochastic Finance. An Introduction in Discrete Time}. Fourth Edition.  {Walter de Gruyter, Berlin}.
 

 
%  \bibitem[Hadar and Russell, 1969]{HR69}
% Hadar, J. and Russell, W.~R. (1969).
% \newblock Rules for ordering uncertain prospects.
% \newblock {\em The American Economic Review}, 59(1):25–34.

% \bibitem[Hadar and Russell, 1971]{HR71}
% Hadar, J. and Russell, W.~R. (1971).
% \newblock Stochastic dominance and diversification.
% \newblock {\em Journal of Economic Theory}, 3(3):288--305.
 

%\bibitem[\protect\citeauthoryear{Hardy et al.}{Hardy et al.}{1934}]{HLP34}
%Hardy, G. H., Littlewood, J. E. and P\'olya, G  (1934). \emph{Inequalities}. Cambridge University Press.


\bibitem[Hofert and W{\"u}thrich, 2012]{hofert2012statistical}
Hofert, M. and W{\"u}thrich, M.~V. (2012).
\newblock Statistical review of nuclear power accidents.
\newblock {\em Asia-Pacific Journal of Risk and Insurance}, 7(1), Article 1.




%\bibitem[\protect\citeauthoryear{Huang et al.}{2020}]{HTZ20}
%    Huang, R. J., Tzeng, L. Y.  and  Zhao, L. (2020).  Fractional degree stochastic dominance. \emph{Management Science},
%  {66}(10): 4630--4647.


  \bibitem[Ibragimov, 2005]{I05}
		Ibragimov, R. (2005).
		\newblock {New majorization theory in economics and martingale convergence
			results in econometrics}.  
		\newblock Ph.D. dissertation, Yale University, New Haven, CT.
 
 \bibitem[Ibragimov et~al., 2015]{ibragimov2015heavy}
Ibragimov, M., Ibragimov, R. and Walden, J. (2015).
\newblock {\em Heavy-Tailed Distributions and Robustness in Economics and
  Finance}, Vol. 214 of Lecture Notes in Statistics, Springer.
  
%     \bibitem[Ibragimov, 2005]{ibragimov2005new}
%		Ibragimov, R. (2005).
%		\newblock {New majorization theory in economics and martingale convergence
%			results in econometrics}.  
%		\newblock Ph.D. dissertation, Yale University, New Haven, CT.

%\bibitem[Ibragimov, 2009]{ibragimov2009portfolio}
%Ibragimov, R. (2009).
%\newblock Portfolio diversification and value at risk under thick-tailedness.
%\newblock {\em Quantitative Finance}, 9(5):565--580.


\bibitem[\protect\citeauthoryear{Ibragimov et al.}{2009}]{IJW09}
Ibragimov, R., Jaffee, D. and Walden, J. (2009). Non-diversification traps in markets for catastrophic risk. \emph{Review of Financial Studies}, 22:959--993.

\bibitem[\protect\citeauthoryear{Ibragimov et al.}{2011}]{IJW11} Ibragimov, R., Jaffee, D. and Walden, J. (2011). Diversification disasters. {\it Journal of Financial Economics}, {99}(2):333--348.

%\bibitem[Ibragimov and Prokhorov, 2017]{ibragimov2017heavy}
%Ibragimov, R. and Prokhorov, A. (2017).
%\newblock {\em Heavy Tails and Copulas: Topics in Dependence Modelling in
%  Economics and Finance}.
%\newblock World Scientific.


\bibitem[\protect\citeauthoryear{Ibragimov and Walden}{2007}]{IW07}  Ibragimov, R. and Walden, J. (2007). The limits of diversification when losses may be large. {\it Journal of Banking and Finance}, 31(8):2551--2569. 

%\bibitem[Ibragimov and Walden, 2008]{ibragimov2008portfolio}
%Ibragimov, R. and Walden, J. (2008).
%\newblock Portfolio diversification under local and moderate deviations from
%  power laws.
%\newblock {\em Insurance: Mathematics and Economics}, 42(2):594--599. 

%\bibitem[Ibragimov and Walden, 2010]{ibragimov2010optimal}
%Ibragimov, R. and Walden, J. (2010).
%\newblock Optimal bundling strategies under heavy-tailed valuations.
%\newblock {\em Management Science}, 56(11):1963--1976.


 
\bibitem[\protect\citeauthoryear{Joag-Dev and Proschan}{Joag-Dev and Proschan}{1983}]{JP83}
Joag-Dev, K. and Proschan, F. (1983). Negative association of random variables with applications. \emph{Annals of Statistics},  11(1):286--295.

 


%\bibitem[\protect\citeauthoryear{Jarrow}{Jarrow}{1986}]{J86}
%Jarrow, R. (1986). The relationship between arbitrage and first order stochastic dominance. {\it Journal of Finance}, 41(4):915--921.


%\bibitem[\protect\citeauthoryear{Kaas et al.}{Kaas et al.}{2004}]{KGT04}
%Kaas, R., Goovaerts, M. and Tang, Q. (2004). Some useful counterexamples regarding comonotonicity. {\it Belgian Actuarial Bulletin}, 4(1):1--4.

%\bibitem[Kley et~al., 2016]{kley2016risk}
%Kley, O., Kl{\"u}ppelberg, C. and Reinert, G. (2016).
%\newblock Risk in a large claims insurance market with bipartite graph
%  structure.
%\newblock {\em Operations Research}, 64(5):1159--1176.


\bibitem[\protect\citeauthoryear{Klugman et al.}{Klugman et al.}{2012}]{KPW12}
Klugman, S. A., Panjer, H. H. and Willmot, G. E. (2012). \emph{Loss Models: From Data to Decisions.} 4th Edition. John Wiley \& Sons.

\bibitem[\protect\citeauthoryear{Lehmann}{Lehmann}{1966}]{L66}
Lehmann, E. L. (1966). Some concepts of dependence. \emph{Annals of Mathematical Statistics},  {37}(5):1137--1153.

%\bibitem[Levy, 1992]{levy1992stochastic}
%Levy, H. (1992).
%\newblock Stochastic dominance and expected utility: Survey and analysis.
%\newblock {\em Management Science}, 38(4):555--593.

% \bibitem[\protect\citeauthoryear{Levy}{Levy}{2016}]{L16} Levy, H.~(2016). \emph{Stochastic Dominance. Investment Decision Making under Uncertainty}. Third Edition.  {Springer}.

%\bibitem[Malinvaud, 1972]{M72}
%Malinvaud, E. (1972).
%The allocation of individual risks in large markets.
%\newblock {\em Journal of Economic Theory}, 4(2):312–328.

\bibitem[Mainik and R{\"u}schendorf, 2010]{mainik2010optimal}
Mainik, G. and R{\"u}schendorf, L. (2010).
\newblock On optimal portfolio diversification with respect to extreme risks.
\newblock {\em Finance and Stochastics}, 14:593--623.


\bibitem[\protect\citeauthoryear{Markowitz}{Markowitz}{1952}]{M52}
Markowitz, H. (1952). The utility of wealth. \emph{Journal of Political Economy}, {60}(2):151--158.


%   \bibitem[\protect\citeauthoryear{Marshall et~al.}{Marshall et~al.}{2011}]{MOA11}
%Marshall, A.~W., Olkin, I. and Arnold, B. (2011).
% {\em Inequalities: Theory of Majorization and Its Applications}.
%  Springer, 2nd edition.


\bibitem[McNeil et~al., 2015]{MFE15}
McNeil, A.~J., Frey, R. and Embrechts, P. (2015).
\newblock {\em Quantitative Risk Management: Concepts, Techniques and
  Tools}. Revised Edition.
\newblock Princeton University Press.



\bibitem[Moscadelli, 2004]{moscadelli2004modelling}
Moscadelli, M. (2004).
\newblock The modelling of operational risk: Experience with the analysis of
  the data collected by the Basel committee. Technical Report 517. \emph{SSRN}: 557214.


  \bibitem[M{\"u}ller, 2024]{M24}
M{\"u}ller, A. (2024).
\newblock Some remarks on the effect of risk sharing and diversification for
  infinite mean risks.
  \newblock {\em  arXiv:2411.10139}.
  


%\bibitem[\protect\citeauthoryear{M\"uller et al.}{2017}]{MSTW17}
%M\"uller, A., Scarsini, M., Tsetlin, I. and Winkler. R. L. (2017).
%Between first and second-order stochastic dominance. \emph{Management Science}, {63}(9):2933--2947.


% \bibitem[M\"uller and Stoyan, 2002]{MS02} {M\"uller, A. and Stoyan, D.} (2002). \newblock \emph{Comparison Methods for Stochastic Models and Risks}. \newblock Wiley, England.

\bibitem[\protect\citeauthoryear{Nelsen}{Nelsen}{2006}]{N06} {Nelsen, R.} (2006). \emph{An Introduction to Copulas}.  Second
Edition. Springer, New York. 


\bibitem[Ne{\v{s}}lehov{\'a} et~al., 2006]{NEC06}
Ne{\v{s}}lehov{\'a}, J., Embrechts, P. and Chavez-Demoulin, V. (2006).
\newblock Infinite mean models and the LDA for operational risk.
\newblock {\em Journal of Operational Risk}, 1(1):3--25.



\bibitem[Nordhaus, 2009]{nordhaus2009analysis}
Nordhaus, W.~D. (2009).
\newblock An analysis of the Dismal Theorem, Yale University: Cowles Foundation Discussion Paper 1686.

% \bibitem[Norton et~al., 2021]{norton2021calculating}
% Norton, M., Khokhlov, V. and Uryasev, S. (2021).
% \newblock Calculating CVaR and bPOE for common probability distributions with
%   application to portfolio optimization and density estimation.
% \newblock {\em Annals of Operations Research}, 299:1281--1315.



\bibitem[\protect\citeauthoryear{OECD}{OECD}{2018}]{O18} {OECD.} (2018). \emph{The Contribution of Reinsurance Markets to Managing Catastrophe Risk}. Available at \url{www.oecd.org}.
%/finance/the-contribution-of-reinsurance-markets-to-managing-catastrophe-risk.pdf}.

%\bibitem[\protect\citeauthoryear{Sarkar}{Sarkar}{2008}]{sarkar2008simes}
%Sarkar, S.~K.  (2008).
% On the Simes inequality and its generalization.
%  In {\em Beyond Parametrics in Interdisciplinary Research: Festschrift
%  in Honor of Professor Pranab K. Sen}  (pp.\ 231--242). Institute of
%  Mathematical Statistics.


\bibitem[Pickands, 1975]{P75}
Pickands, J. (1975).  
\newblock 
Statistical inference using extreme order statistics.
\newblock {\em Annals of Statistics}, 3:119--131.

%\bibitem[Proschan, 1965]{proschan1965peakedness}
%Proschan, F. (1965).
%\newblock Peakedness of distributions of convex combinations.
%\newblock {\em Annals of Mathematical Statistics}, 36(6):1703--1706.


% \bibitem[Quirk and Saposnik, 1962]{quirk1962admissibility}
% Quirk, J.~P. and Saposnik, R. (1962).
% \newblock Admissibility and measurable utility functions.
% \newblock {\em The Review of Economic Studies}, 29(2):140--146.



\bibitem[Rizzo, 2009]{rizzo2009new}
Rizzo, M.~L. (2009).
\newblock New goodness-of-fit tests for Pareto distributions.
\newblock {\em ASTIN Bulletin}, 39(2):691--715.

%\bibitem[\protect\citeauthoryear{Rothschild and Stiglitz}{1970}]{RS70}
%Rothschild, M. and Stiglitz, J. E. (1970). Increasing risk: I. A definition. \emph{Journal of Economic Theory},  {2}(3):225--243.




  \bibitem[Samuelson, 1967]{S67}
Samuelson, P. A. (1967). General proof that diversification pays. \emph{Journal of Financial and Quantitative Analysis}, {2}(1):1--13.


%    \bibitem[\protect\citeauthoryear{Schmidt and Zank}{Schmidt and Zank}{2008}]{SZ08}
%Schmidt, U. and Zank, H. (2008). Risk aversion in cumulative prospect theory. \emph{Management Science}, {54}:208--216.





  \bibitem[Shaked and Shanthikumar, 2007]{SS07}
Shaked, M. and Shanthikumar, J.~G. (2007).
\newblock {\em Stochastic Orders}.
\newblock Springer.

 % \bibitem[Shapiro et~al., 2021]{shapiro2021lectures}
 % Shapiro, A., Dentcheva, D. and Ruszczynski, A. (2021).
 % \newblock {\em Lectures on Stochastic Programming: Modeling and Theory}.
 % \newblock SIAM.

%\bibitem[Shapiro and Varian, 1999]{shapiro1999information}
%Shapiro, C. and Varian, H.~R. (1999).
%\newblock {\em Information Rules: A Strategic Guide to the Network Economy}.
%\newblock Harvard Business Press.

\bibitem[Silverberg and Verspagen, 2007]{silverberg2007size}
Silverberg, G. and Verspagen, B. (2007).
\newblock The size distribution of innovations revisited: An application of
  extreme value statistics to citation and value measures of patent
  significance.
\newblock {\em Journal of Econometrics}, 139(2):318--339.



\bibitem[Sornette et~al., 2013]{sornette2013exploring}
Sornette, D., Maillart, T. and Kr{\"o}ger, W. (2013).
\newblock Exploring the limits of safety analysis in complex technological
  systems.
\newblock {\em International Journal of Disaster Risk Reduction}, 6:59--66.







%\bibitem[\protect\citeauthoryear{Simes}{Simes}{1986}]{S86}
%Simes, R. J. (1986).
 % An improved {B}onferroni procedure for multiple tests of
%  significance.
 % {\em Biometrika}, 73(3), 751--754.



%  \bibitem[Taleb, 2020]{T20}
%Taleb, N. (2020).
%\emph{Statistical Consequences of Fat Tails}. 
% STEM Academic Press. 

 

 

  \bibitem[Tasche, 2000]{tasche2000conditional}
Tasche, D. (2000).
 Conditional expectation as quantile derivative.
 {\em arXiv}: 0104190.
 
\bibitem[\protect\citeauthoryear{Tversky and Kahneman}{1992}]{TK92}
Tversky, A. and Kahneman, D. (1992). Advances in prospect theory: Cumulative representation of Uncertainty. \emph{Journal of Risk and Uncertainty}, {5}(4): 297--323.






%\bibitem[Wang, 1996]{W96}
%Wang, S. (1996).
%\newblock Premium calculation by transforming the layer premium density.
%\newblock {\em ASTIN Bulletin: The Journal of the IAA}, 26(1):71--92.




 

 
% \bibitem[Wheatley et~al., 2016]{wheatley2016reassessing}
% Wheatley, S., Sovacool, B.~K. and Sornette, D. (2016).
% \newblock Reassessing the safety of nuclear power.
% \newblock {\em Energy Research \& Social Science}, 15:96--100.




\bibitem[Wang et~al., 2020]{wang2020distortion}
Wang, Q., Wang, R. and Wei, Y. (2020).
\newblock Distortion riskmetrics on general spaces.
\newblock {\em ASTIN Bulletin}, 50(3):827--851.



\bibitem[\protect\citeauthoryear{Weitzman}{Weitzman}{2009}]{W09}
Weitzman, M. L. (2009). On modeling and interpreting the economics of catastrophic climate change. \emph{Review of Economics and Statistics}, 91(1):1-19.

%\bibitem[Whitmore, 1970]{whitmore1970third}
%Whitmore, G.~A. (1970).
%\newblock Third-degree stochastic dominance.
%\newblock {\em The American Economic Review}, 60(3):457--459.

\bibitem[\protect\citeauthoryear{Yaari}{Yaari}{1987}]{Y87}
{Yaari, M. E.} (1987). The dual theory of choice under risk. \emph{Econometrica}, {55}(1):95--115.


\bibitem[Zhu et~al., 2023]{zhu2023asymptotic}
Zhu, W., Li, L., Yang, J., Xie, J. and Sun, L. (2023).
 Asymptotic subadditivity/superadditivity of Value-at-Risk under tail
  dependence.
\newblock {\em Mathematical Finance}, 33(4):1314--1369.






 
 





 









\end{thebibliography}
\end{document}